\documentclass[10pt,sigconf,balance=true]{acmart}

\usepackage{blindtext}
\usepackage[utf8]{inputenc}
\usepackage{epsfig,xspace,url,xcolor}
\usepackage{subcaption}
\usepackage{amsmath}
\usepackage{mathtools}
\usepackage{amsthm}
\usepackage{tikz}
\usepackage{caption}
\usepackage{graphicx}
\usepackage[ruled, lined, linesnumbered, commentsnumbered, longend]{algorithm2e}
\usepackage[noend]{algpseudocode}
\usepackage{url}
\usepackage{colortbl}
\usepackage{xcolor}

\usepackage{thm-restate}

\usepackage{wrapfig}

\renewcommand\footnotetextcopyrightpermission[1]{} % removes footnote with conference info
\input{macros}

\let\ACMmaketitle=\maketitle
\renewcommand{\maketitle}{\begingroup\let\footnote=\thanks \ACMmaketitle\endgroup}

\makeatletter
\def\author@bx@sep{0pc}
\makeatother

\begin{document}

% \sloppy
% \newcommand{\name}{\textsc{No Name}\xspace}
% \title{Understanding the Throughput of Demand-Aware Reconfigurable Datacenter Networks}
\title{Understanding the Throughput Bounds of Reconfigurable~Datacenter Networks}
% \title{Understanding the Throughput of Self-Adjusting~Networks}
% \title{Throughput of Self-Adjusting Networks}

% \author{Paper \#1157, \pageref{bodyLastPage} pages body, \pageref{LastPage} pages total}
\author{Vamsi Addanki}
\affiliation{
  \institution{TU Berlin}
  % \city{Berlin}
  % \country{Germany}
}
\author{Chen Avin}
\affiliation{
  \institution{Ben-Gurion University of the Negev}
  % \city{Berlin}
  % \country{Germany}
}
\author{Stefan Schmid}
\affiliation{
  \institution{TU Berlin}
  % \city{Berlin}
  % \country{Germany}
}
\renewcommand{\shortauthors}{X et al.}

% \begin{abstract} 
% \end{abstract}

\sloppy
\begin{abstract}
    % Datacenters witnessed an explosive growth in the overall network traffic over the past decade. 
    % With the recent entry of GPU training workloads in the datacenter, the network traffic is only expected to increase at an even faster rate. 
    % Unfortunately, network infrastructure is unable to keep up with the increasing network demand due to moore's law limitations for capacity scaling. 
    % In the recent past, researchers have come up with novel technologies based on reconfigurable optical circuit switches in order to bypass moore's law for better capacity scaling. 

    The increasing gap between the growth of datacenter traffic volume and the capacity of electrical switches led to the emergence of reconfigurable datacenter network designs based on optical circuit switching. A multitude of research works, ranging from demand-oblivious (e.g., RotorNet, Sirius) to demand-aware (e.g., Helios, ProjecToR) reconfigurable networks, demonstrate significant performance benefits.
    Unfortunately, little is formally known about the achievable throughput of such networks. Only recently have the throughput bounds of demand-oblivious networks been studied.
    In this paper, we tackle a fundamental question: \textit{Whether and to what extent can demand-aware reconfigurable networks improve the throughput of datacenters?}

    This paper attempts to understand the landscape of the throughput bounds of reconfigurable datacenter networks. Given the rise of machine learning workloads and collective communication in modern datacenters, we specifically focus on their typical communication patterns, namely uniform-residual demand matrices.
    We formally establish a separation bound of demand-aware networks over demand-oblivious networks, proving analytically that the former can provide at least $16\%$ higher throughput.
    Our analysis further uncovers new design opportunities based on periodic, fixed-duration reconfigurations that can harness the throughput benefits of demand-aware networks while inheriting the simplicity and low reconfiguration overheads of demand-oblivious networks.
    Finally, our evaluations corroborate the theoretical results of this paper, demonstrating that demand-aware networks significantly outperform oblivious networks in terms of throughput.
    This work barely scratches the surface and unveils several intriguing open questions, which we discuss at the end of this paper.
\end{abstract}

\maketitle
\thispagestyle{plain}
\pagestyle{plain}

\section{Introduction}
\label{sec:introduction}

\begin{figure}[t]
    \centering
    \includegraphics[width=1\linewidth]{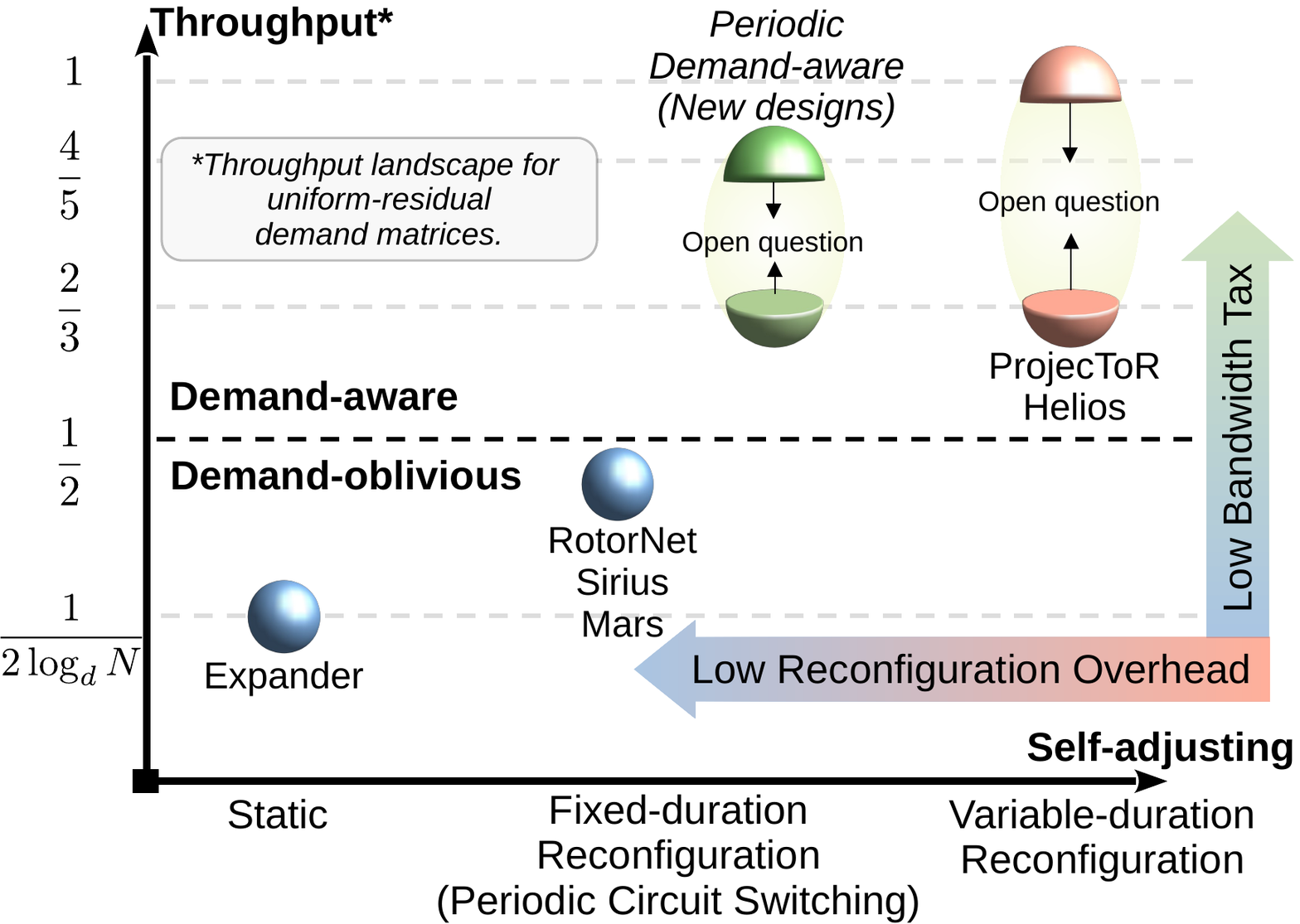}
    \caption{The landscape of throughput bounds for reconfigurable datacenter networks under uniformly skewed communication patterns: while prior works show a tight bound close to $\frac{1}{2}$ for demand-oblivious networks, we show the first separation result \ie demand-aware networks are strictly better in terms of throughput. Even simple demand-aware networks based on periodic fixed-duration reconfigurations (similar to RotorNet \& Sirius) can achieve at least $16\%$ better throughput in the worst-case.}
    % \vspace{-4mm}
    \label{fig:intro}
\end{figure}

Datacenters have experienced explosive growth in overall network traffic volume over the past decade~\cite{10.1145/2785956.2787508}.
With the recent introduction of high-bandwidth Machine Learning workloads into datacenters, the peak network traffic is expected to increase even more rapidly~\cite{10.1145/3544216.3544265}.
Unfortunately, traditional networks, which are built using electrical packet switches, struggle to keep up with this growing demand~\cite{10.1145/3387514.3406221}. Further, the rapid evolution of datacenter applications and their changing bandwidth requirements implies: ``the best laid plans quickly become outdated and inefficient, making incremental and adaptive evolution a necessity''~\cite{10.1145/3544216.3544265}.
This led to the emergence of novel technologies based on reconfigurable optical circuit switches~\cite{10.1145/3098822.3098838,10.1145/3387514.3406221,10.1145/2934872.2934911,10.1145/1851182.1851223,10.1145/3544216.3544265}.

Two prominent types of reconfigurable datacenter networks emerged in the recent past: demand-oblivious~\cite{10.1145/3098822.3098838,10.1145/3387514.3406221,10.1145/3579312} and demand-aware~\cite{10.1145/2934872.2934911,10.1145/1851182.1851223,10.1145/3579449}. Notably, these networks are optically circuit-switched and feature bufferless switches.
Demand-oblivious networks, such as RotorNet~\cite{10.1145/3098822.3098838}, Sirius~\cite{10.1145/3387514.3406221} and Opera~\cite{opera}, provide low reconfiguration overheads (in the order of nanoseconds) but compromise on throughput due to their fixed and periodic switching schedules, which are independent of underlying communication patterns. In contrast, Demand-aware networks like Helios~\cite{10.1145/1851182.1851223}, ProjecToR~\cite{10.1145/2934872.2934911}, and TopoOpt~\cite{285119} achieve higher throughput because their switching schedules are optimized for the underlying communication patterns, albeit at the cost of high reconfiguration delays due to complex control plane mechanisms. A vast majority of the existing literature relies on empirical evaluations to study the performance benefits of such networks.
Analytical insights into the throughput bounds of demand-oblivious networks have only recently been obtained~\cite{10.1145/3579312,10.1145/3519935.3520020,10.1145/3491050}.
Unfortunately, formal knowledge regarding the achievable throughput of demand-aware reconfigurable datacenter networks and their comparative performance against demand-oblivious networks remains limited.

Common wisdom suggests that a demand-oblivious network is as effective as a demand-aware network if the underlying communication patterns change rapidly and are unpredictable. However, the increasing prevalence of Machine Learning workloads in modern datacenters challenges this notion, advocating for demand-aware network designs. Specifically, Machine Learning workloads, particularly DNN training workloads, tend to be periodic, with the demand matrix remaining constant throughout the training duration~\cite{285119}. For example, collective communications such as ring-AllReduce produce a characteristic permutation demand matrix.
These types of communications often lead to uniformly skewed demand matrices \ie the communication patterns are skewed but largely similar across all nodes in the network. This uniformity presents an opportunity to optimize the network topology to better align with the underlying communication patterns, as demonstrated by solutions like TopoOpt~\cite{285119}. This leads to the pivotal question:

\medskip
\emph{Can demand-aware networks consistently achieve better throughput than demand-oblivious networks, and if so, to what extent?}
\medskip

Figure~\ref{fig:intro} summarizes our results, illustrating the landscape of throughput bounds for demand-aware and demand-oblivious networks under uniform-residual demand matrices. Prior works have demonstrated that the throughput of demand-oblivious networks is tightly bounded by $\frac{1}{2}$~\cite{10.1145/3579312,10.1145/3519935.3520020}. In this work, we formally establish the first separation result: demand-aware networks are strictly superior to demand-oblivious networks in terms of throughput under uniform-residual demand matrices. Specifically, we show that the throughput of a demand-aware network is at least $\frac{2}{3}$ \ie at least $16\%$ greater than that of a demand-oblivious network.

Our analysis of throughput involves a novel technique to decompose the demand matrix in to a floor matrix and a residual matrix. Intuitively, the floor matrix represents the portion of the demand matrix that allows for easy optimization of the topology. Specifically, any source-destination demand in the floor matrix can be routed in a single hop by adding corresponding direct links without incurring any bandwidth tax. The residual matrix corresponds to the portion of the demand matrix for which an oblivious or static topology performs reasonably well. Leveraging these insights, we establish both lower and upper bounds for the throughput of demand-aware networks.

Striking a balance between the throughput benefits of demand-aware networks, and the low reconfiguration delays (overhead) of demand-oblivious networks, we uncover interesting new demand-aware designs based on periodic fixed-duration reconfigurations. Even for such designs, we formally show that the throughput lower bound of $\frac{2}{3}$ holds. Interestingly, our analysis reveals that these demand-aware periodic networks cannot achieve a throughput greater than $\frac{4}{5}$, which is the upper bound \ie up to a $30\%$ increase in throughput compared to demand-oblivious networks, achieved without sacrificing the simplicity of circuit switching and while maintaining low reconfiguration delays.

Our evaluations, based on linear programming approach, corroborate our theoretical findings by demonstrating notable improvements in throughput (maximum sustained load) for demand-aware network designs when compared to demand-oblivious networks like RotorNet and Sirius. Specifically, we observe that demand-aware periodic networks can increase throughput by as much as $49\%$ over demand-oblivious networks, achieve up to a $2.4$-fold improvement over static networks, and offer a $30\%$ absolute improvement in throughput in the worst-case scenario.

We view our work just as a first step towards understanding the throughput benefits of reconfigurable datacenter networks. We discuss various interesting open questions and future research directions at the end of this paper. For instance, while our throughput bounds apply to uniform-residual demand matrices, it remains an open question in theoretical research whether the landscape might differ with the consideration of other types of matrices. Additionally, our new demand-aware periodic designs demonstrate promising theoretical throughput properties. Yet, there is an open question in systems research regarding the feasibility of adapting the switching schedules of networks like RotorNet and Sirius in real-time, potentially even less frequently.

\medskip
In summary, our main contributions in this paper are:

\begin{itemize}[leftmargin=*,label=\small{\textcolor{myred}{$\blacksquare$}}]
    \item A first separation result proving that demand-aware reconfigurable datacenter networks are strictly superior to demand-oblivious networks under uniform-residual demand matrices.
    \item Innovative, yet simple, demand-aware network designs based on periodic fixed-duration reconfigurations. These networks achieve a throughput of at least $\frac{2}{3}$ (lower bound) and at most $\frac{4}{5}$ (upper bound) under uniform-residual demand matrices: a significant improvement over demand-oblivious networks.
    \item Empirical evaluations that support our theoretical findings, highlighting the throughput benefits of demand-aware networks compared to their demand-oblivious counterparts.
    \item As a contribution to the research community, and to facilitate future research work, all our artefacts and source code will be made available online together with this paper.
\end{itemize}

% \medskip
% \emph{Whether and to what extent demand-aware networks can improve the throughput of datacenter networks?}
% \medskip

% Our answer is partially yes: demand-aware networks are strictly better in terms of throughput compared to demand-oblivious networks under uniformly skewed communication patterns. 

\textit{This work does not raise any ethical issues.}

% \clearpage

\section{Motivation}
\label{sec:motivation}
In this section, we first provide a brief background on reconfigurable networks and introduce our definition of throughput (\S\ref{sec:rdcns}).
We then discuss the structure in the demand matrices of emerging machine learning workloads that motivate the need for a better understanding of the achievable throughput of reconfigurable datacenter networks (\S\ref{sec:workloads}). We prove few trivial throughput bounds making a case for demand-aware reconfigurable networks (\S\ref{sec:trivial}). Finally, we discuss the goals of our analysis ahead (\S\ref{sec:roadmap}).

\subsection{Reconfigurable Datacenter Networks}
\label{sec:rdcns}
Optical circuit switches are the fundamental building blocks of reconfigurable datacenter networks. In contrast to electrical packet switches, optical circuit switches are bufferless. Typically, servers are arranged into racks and each rack is connected to a top of the rack (ToR) switch. All the ToR switches are then interconnected by a layer of circuit switches. Depending on the type of reconfigurable network, the functionality of the underlying circuit switches differs (described next). We consider a datacenter with $n$~ToR switches, each with $u$~incoming and outgoing links; and $u$~circuit switches, each with $n$~input and output ports. Figure~\ref{fig:rdcn} illustrates a typical reconfigurable datacenter network.

\medskip
\noindent \textbf{Demand-oblivious reconfigurable networks:} Aiming at fast reconfiguration and low overhead, demand-oblivious networks do rely on control plane to configure the circuit switching schedule. Rather, a predefined set of matchings\footnote{Matching defines the forwarding from each input port to output port \ie light received from an input port is directly forwarded without any processing to an output port based on the matching.} installed on the circuit switches are executed in a \emph{periodic} manner. Further, each circuit switch executes a matching for a fixed amount of time and each switch takes a fixed amount of time to reconfigure to its next matching (fixed-duration reconfiguration). For instance, RotorNet~\cite{10.1145/3098822.3098838} deploys $\frac{n}{u}$ number of matchings in each circuit switch. Each matching is executed for $\Delta$ amount of time and it takes $0.1\times \Delta$ amount of time for each switch to reconfigure to its next matching. In essence, the whole network emulates a complete graph (mesh topology) over time~\cite{10.1145/3579312} \ie every ToR connects to every other ToR over one period of the switching schedule.

\medskip
\noindent \textbf{Demand-aware reconfigurable networks:} Aiming at optimizing the throughput, demand-aware networks rely on control plane to measure the demand matrix and, to compute and configure optimal circuit switching schedules. The resulting switching schedule can be of any length and not necessarily periodic. Demand-aware networks essentially optimize the topology for the underlying communication patter but their dependency on the control plane increases the reconfiguration overhead \ie measuring the demand matrix adds latency and calculating optimal switching schedules is computationally intensive.

\medskip
\noindent \textbf{Demand-aware static networks:} Optical circuit switches in demand-aware static networks serve a similar function as that of a patch-panel or robotic-arm. The control plane reconfigures the matching executed by each switch based on the measured demand matrix. Note that the control plane only configures one matching at each switch \ie the topology remains static for the entire duration until control plane performs another reconfiguration. This type of reconfigurable network has been deployed at Google~\cite{10.1145/3544216.3544265}.

\begin{figure}[t]
    \centering
    \includegraphics[width=1\linewidth]{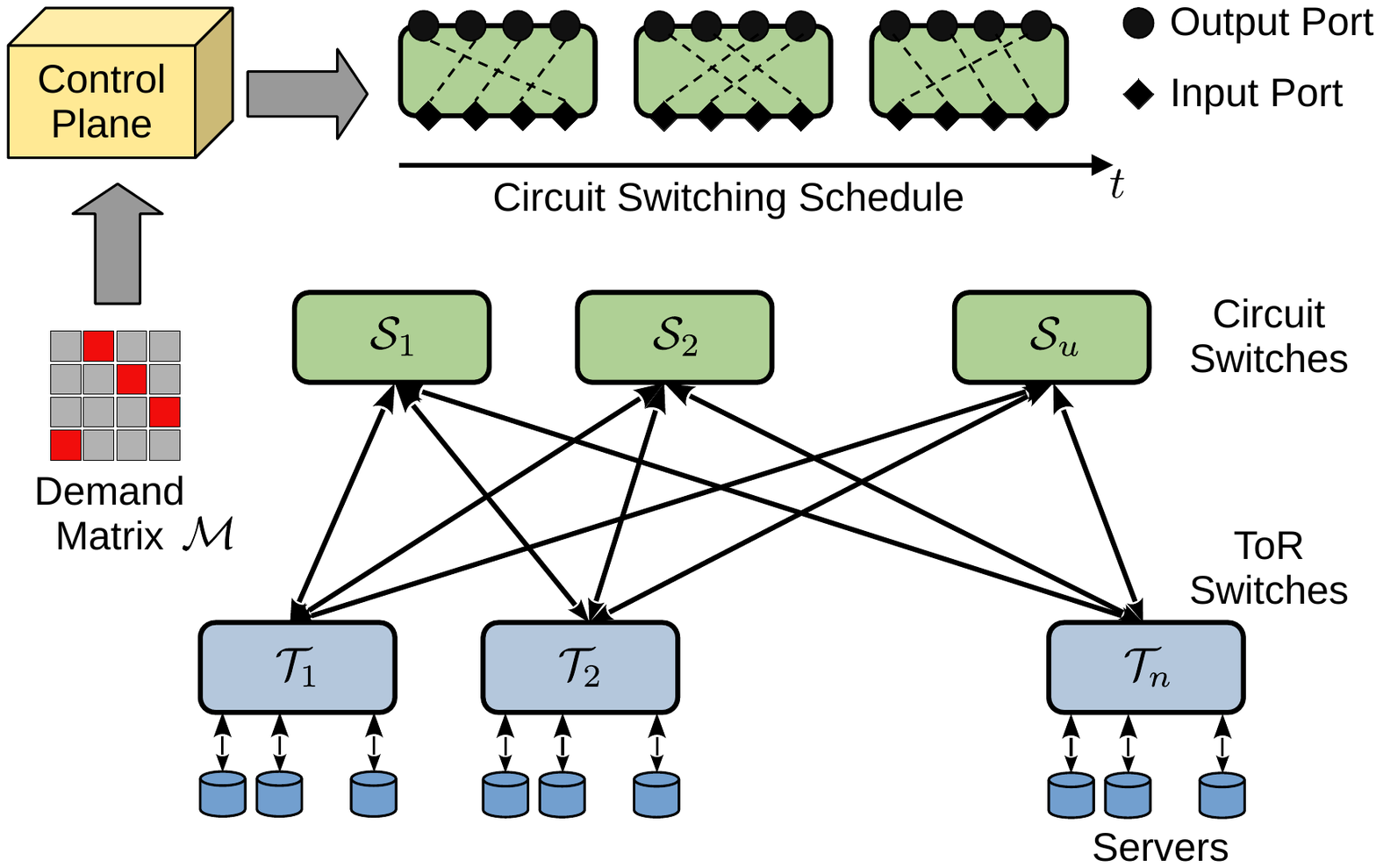}
    \caption{Physical topology of a reconfigurable datacenter network.}
    % \vspace{-4mm}
    \label{fig:rdcn}
\end{figure}

In order to quantify the throughput of each type of network, we first formally define the communication pattern \ie the demand matrix (Definition~\ref{def:demand-matrix}). For simplicity, we assume that the topology is not oversubscribed and aggregate the server to server demand to represent ToR to ToR demand. The demand matrix specifies the demand in bits per second between each pair of ToR switches \ie the total demand originating from a source ToR towards a destination ToR. Following prior work~\cite{10.1145/3452296.3472913}, we consider the hose model~\cite{10.1145/316188.316209} 
such that the total demand originating from (and destined to) each ToR is less than its corresponding capacity limits.
% and specifically the saturated set of demand matrices\footnote{It suffices to consider saturated demand matrices to study throughput~\cite{10.1145/3452296.3472913}.} 

\begin{definition}[Demand matrix]\label{def:demand-matrix}
    Given a set $\mathcal{T}$ of $n$ ToR switches each with $u$ outgoing and incoming links of capacity~$c$, a demand matrix specifies the demand rate between every pair of ToRs in bits per second defined as $\mathcal{M}=\{ m_{u,v} \mid u\in \mathcal{T}, v\in \mathcal{T} \}$ where $m_{u,v}$ is the demand between the pair $u,v$. The demand matrix is such that the total demand originating at a source~$s$ is less than its outgoing capacity and the total demand terminating at a destination~$d$ is less than its incoming capacity \ie $\sum_{u\in V} m_{s,u} \le c\cdot u$ and $\sum_{u\in V} m_{u,d} \le c \cdot u$.
\end{definition}

For a given communication pattern and the corresponding demand matrix (Definition~\ref{def:demand-matrix}), we define throughput as the maximum scaling factor such that there exists a feasible flow that can satisfy the scaled demand subject to flow conservation and capacity constraints. We denote flow by $F: P\mapsto \mathbb{R}^{+}$, a map from the set of all paths $P$ (static or temporal) to the set of non-negative real numbers. This mapping naturally ensures that the flow transmitted from a source eventually reaches the destination along a path $p\in P$. To obey capacity constraints, a feasible mapping is such that the sum of all flows traversing a link do not exceed the link capacity. We are now ready to define throughput formally.

\begin{definition}[Throughput]\label{def:throughput}
    Given a demand matrix $\mathcal{M}$ and a reconfigurable network, throughput denoted by $\theta(\mathcal{M})$ is the highest scaling factor such that there exists a feasible flow for the scaled demand matrix $\theta(\mathcal{M})\cdot \mathcal{M}$. Throughput~$\theta^*$ is the highest scaling factor for a worst-case demand matrix \ie $\theta^* = \underset{\mathcal{M}\in \hat{M}}{\min} \theta(\mathcal{M})$, where $\hat{M}$ is the set of all demand matrices.
\end{definition}

Intuitively, throughput for a specific communication pattern captures the maximum sustainable load by the underlying topology. Based on Definition~\ref{def:throughput}, similar to prior works~\cite{10.1145/3452296.3472913,10.1145/3579312,10.1145/3491050,7877143}, throughput of a topology is the minimum throughput across the set of all saturated demand matrices \ie if a topology has throughput $\theta^*$, then it can achieve at least throughput $\theta^*$ for \textit{any} demand matrix and at most throughput $\theta^*$ for a worst-case demand matrix. In contrast to traditional datacenter networks, the fundamental challenge to study throughput in the context of reconfigurable networks is that the topology changes over time and can be even be a function of the demand matrix in the case of demand-aware networks. In constructing optimal topologies for demand-aware networks, prior works rely on the structure of the underlying demand matrix and use the following intuitions: \first establish demand-aware links between source-destination pairs with large demand in order to minimize path lengths (bandwidth tax) for bulk traffic and \second ensure connectivity to satisfy every source-destination demand. We generalize these intuitions and postulate the following problem \ie Integer-Residual decomposition of a demand matrix. 

\begin{definition}[Integer-residual decomposition ]\label{def:integer-residual-decomp}
    An integer-residual decomposition of a demand matrix $\mathcal{M}$ is two matrices $Int(\M)$, and $Res(\M)$.
    $Int(\M)$ consists of only integer values, where each cell is either a floor or a ceiling of its corresponding cell in $\M$, \ie $Int(\M)_{u,v} = \lfloor \M_{u,v} \rfloor$ or 
    $Int(\M)_{u,v} = \lceil \M_{u,v} \rceil$.
    Additionally, the sum of each row and column in $\mathcal{M}$ is bounded by the corresponding row or column in $\M$,  \ie  for each row $u$, $\sum_v Int(\M)_{u,v} \le \sum_v \M_{u,v}$ and similarly for each column.
    In turn, the residual matrix $Res(M)$ is defined as the reminder in each cell or zero. Formally 
    $Res(\M)_{u,v} = \max \left(\M_{u,v}-Int(\M)_{u,v}, 0 \right)$.
\end{definition}

It follows from the definition that for all $u,v \in \T$,
$Res(\M)_{u,v} < 1$ and $\M_{u,v} \le Int(\M)_{u,v} +Res(\M)_{u,v}$. Additionally, every matrix has an integer-residual decomposition. For example, using floor function for the integer matrix \ie $Int(\M)_{u,v} = \lfloor \M_{u,v} \rfloor$ and $Res(\M)_{u,v}=\M_{u,v} - Int(\M)_{u,v}$ for every $u,v\in \T$. Intuitively, $Int(\M)$ corresponds to bulk portion of the demand matrix for which the topology can be optimized for throughput, whereas, $Res(\M)$ dictates the connectivity requirements that the optimized demand-aware topology must satisfy.

\begin{figure*}
    \centering
    \begin{tabular}{ccccc}
     \raisebox{.6cm}{\rotatebox[origin=lc]{90}{Floor Matrix}}
%     \begin{subfigure}{0.24\linewidth}
%         \centering
        \includegraphics[width=0.24\linewidth]{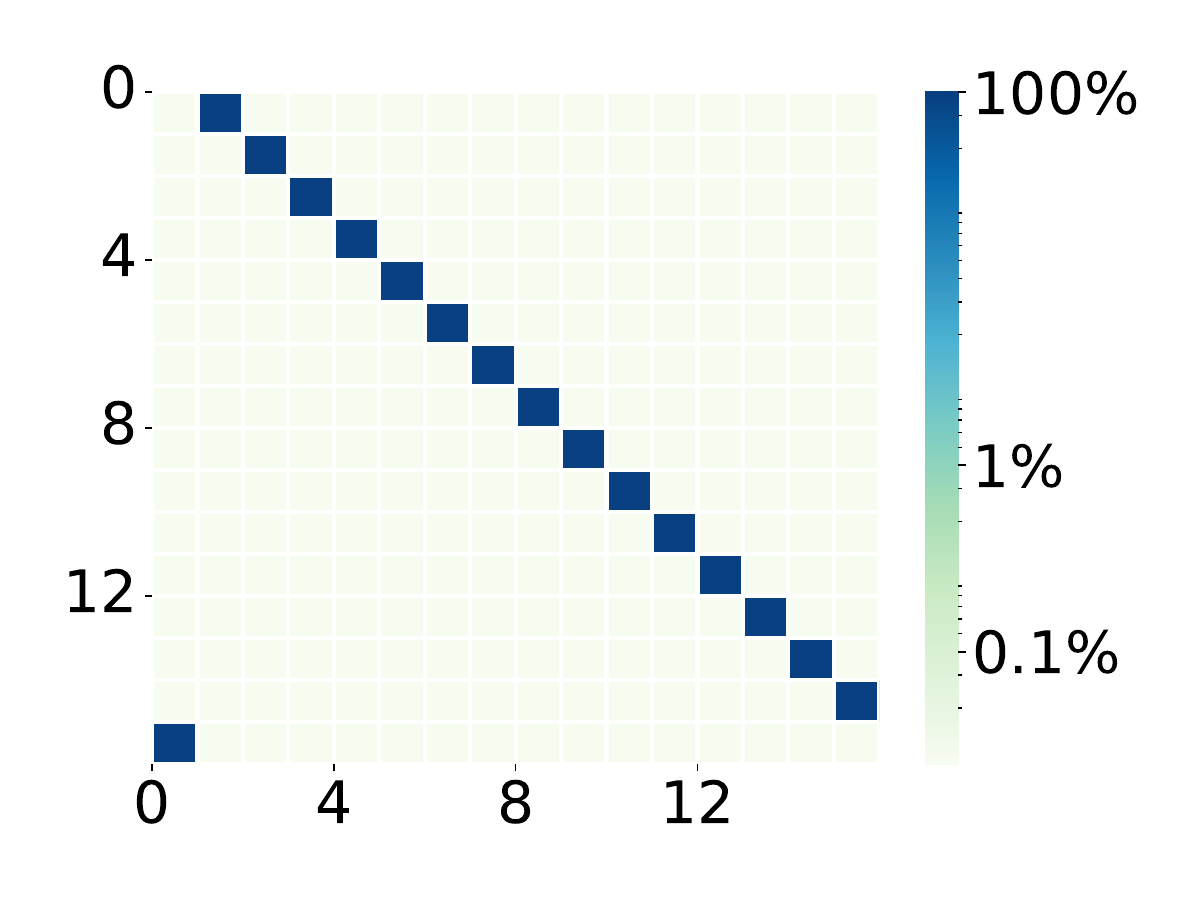}
%         \subcaption{DLRM data parallelism}
        % \vspace{-4mm}
%         \label{fig:floor-data}
%     \end{subfigure} &
%     \begin{subfigure}{0.24\linewidth}
%        \centering
        \includegraphics[width=0.24\linewidth]{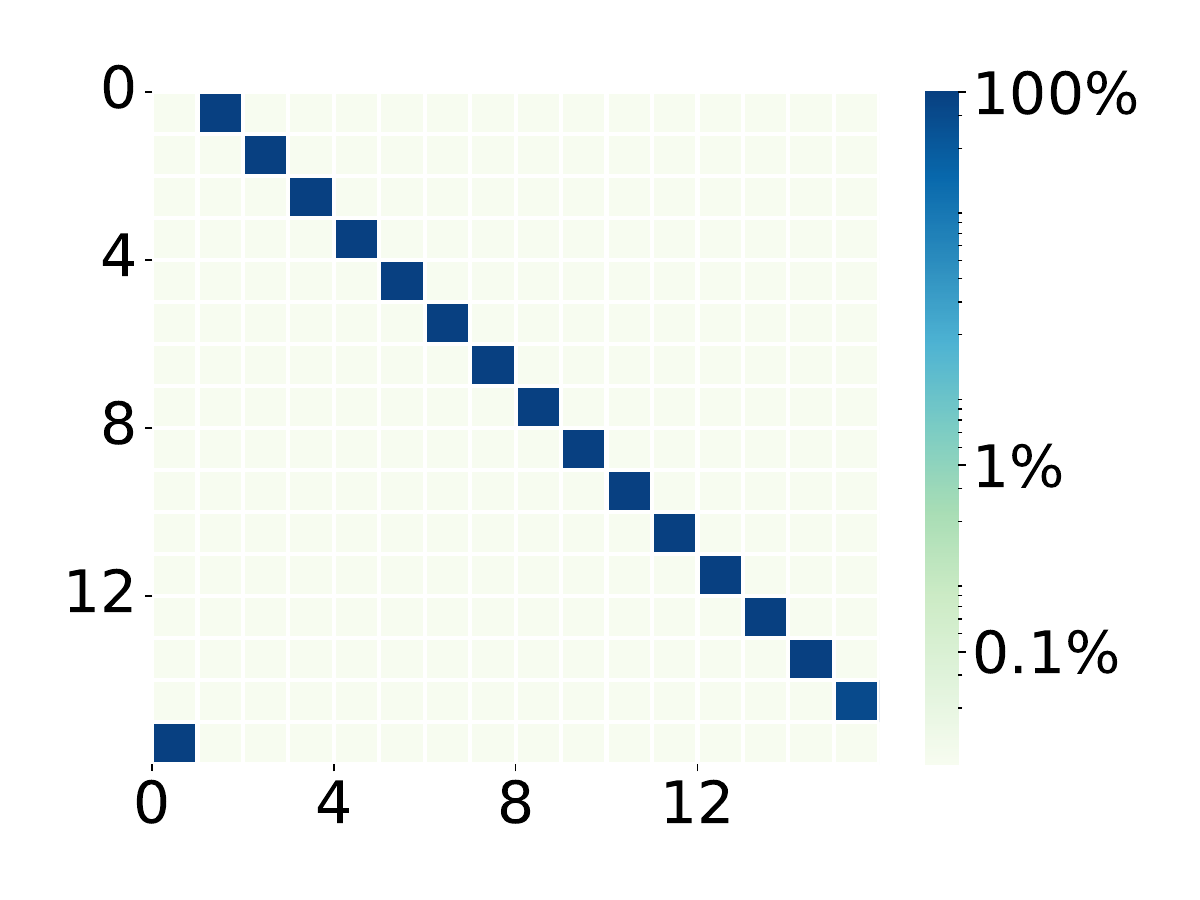}
%         \subcaption{DLRM hybrid parallelism}
%        % \vspace{-4mm}
%         \label{fig:floor-hybrid}
%    \end{subfigure}  &
%     \begin{subfigure}{0.24\linewidth}
%         \centering
        \includegraphics[width=0.24\linewidth]{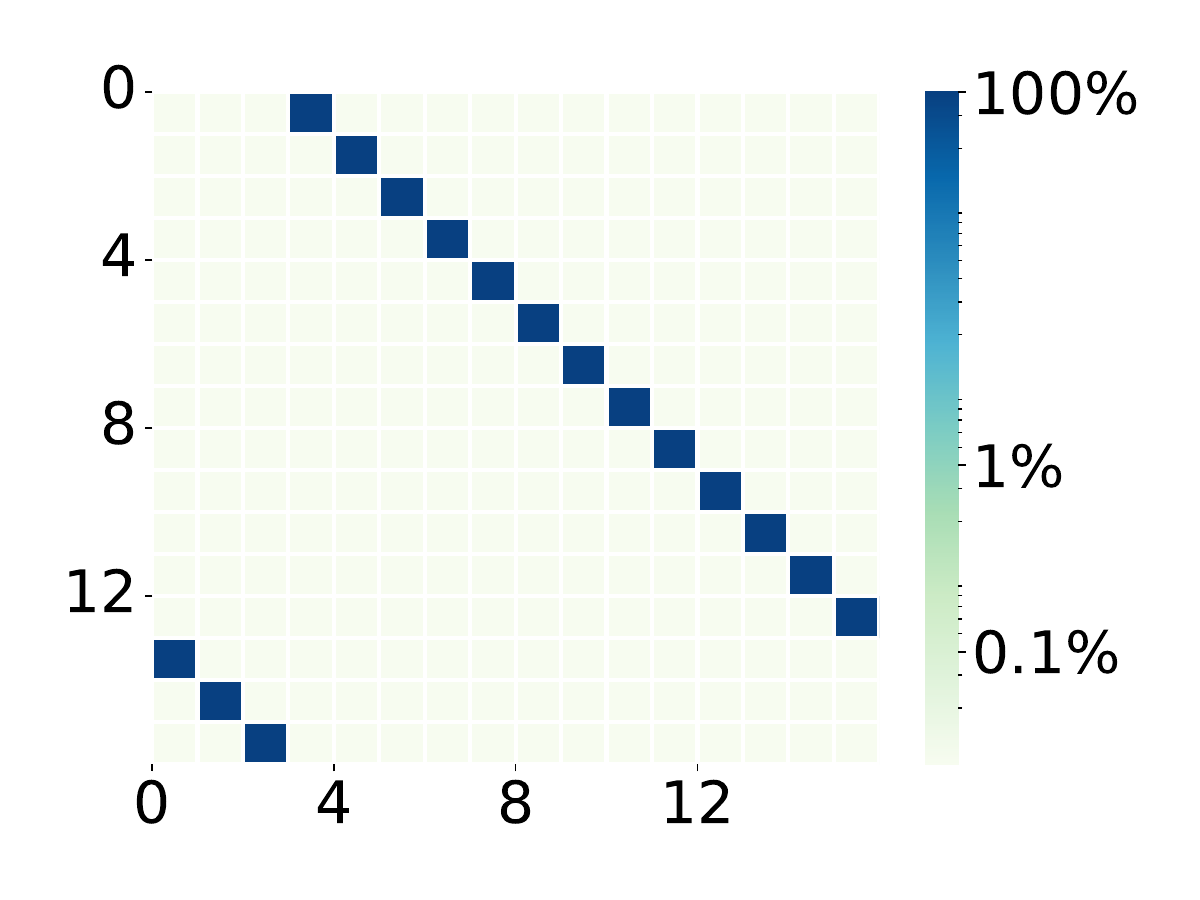} 
%         \subcaption{DLRM $+3$ permuted\\}
%         % \vspace{-4mm}
%         \label{fig:floor-heatmap2}
%     \end{subfigure} &
%     \begin{subfigure}{0.24\linewidth}
%        \centering
        \includegraphics[width=0.24\linewidth]{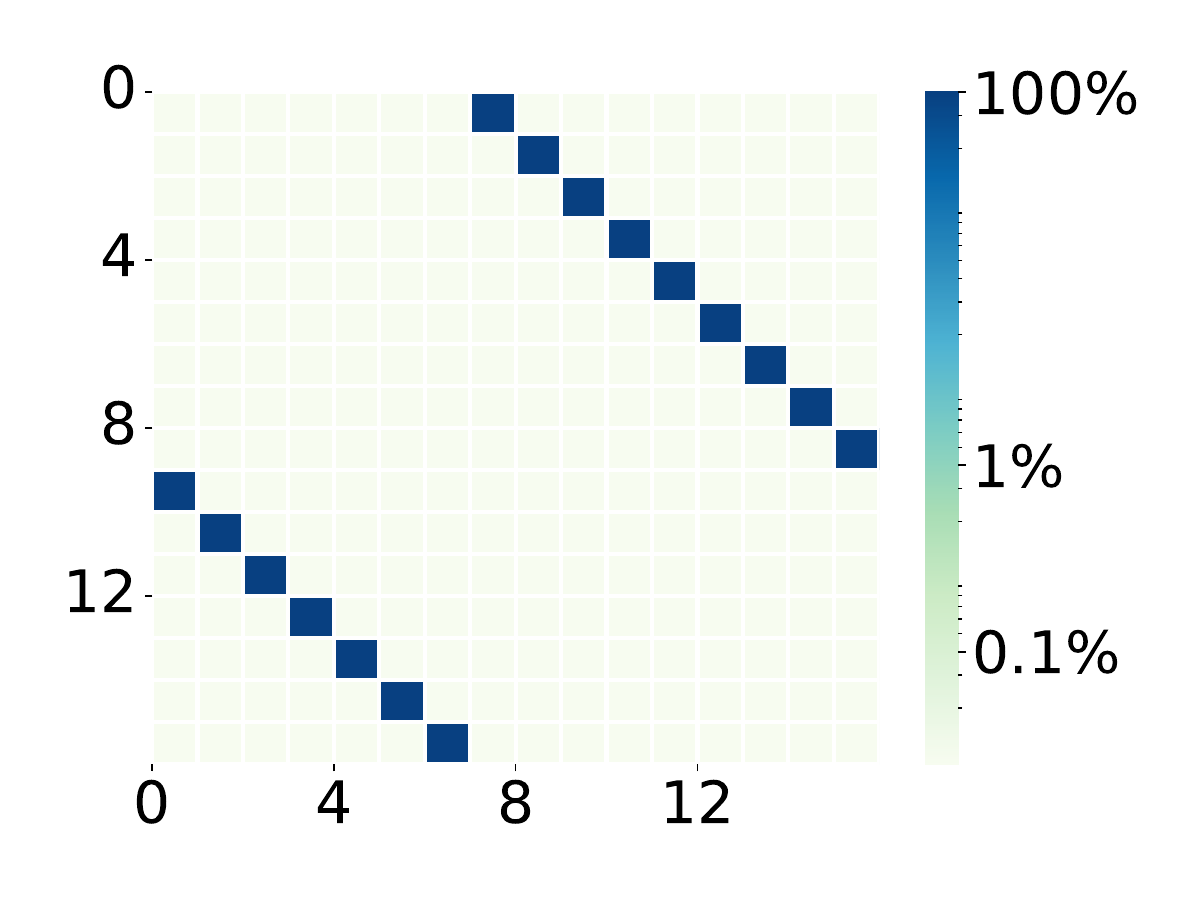} \\
%         \subcaption{DLRM $+7$ permuted\newline}
%         % \vspace{-4mm}
%         \label{fig:floor-heatmap3}
%     \end{subfigure} \\
    \raisebox{.4cm}{\rotatebox[origin=lc]{90}{Residual Matrix}}
%     \begin{subfigure}{0.24\linewidth}
%         \centering
         \includegraphics[width=0.24\linewidth]{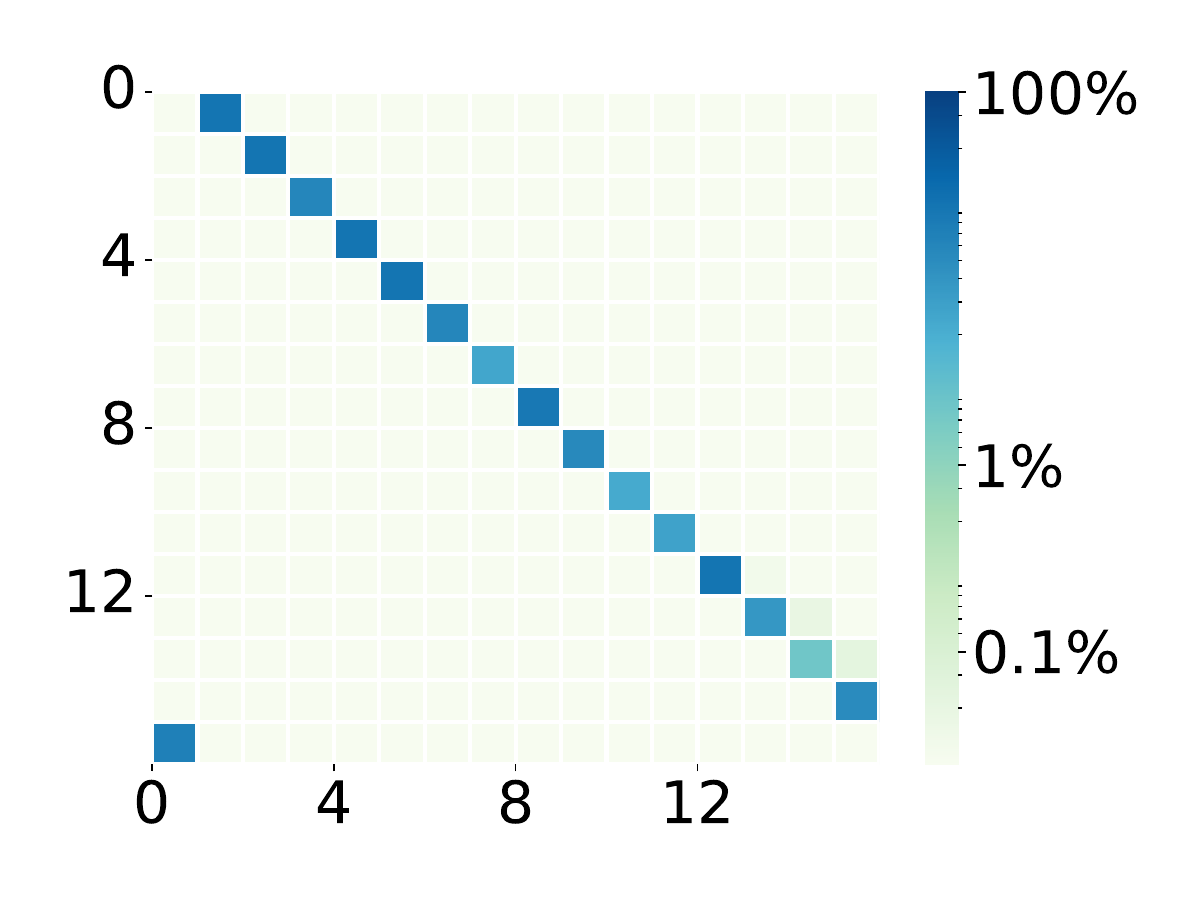} 
%         \subcaption{DLRM data parallelism}
%         % \vspace{-4mm}
%         \label{fig:residual-data}
%     \end{subfigure} &
%     \begin{subfigure}{0.24\linewidth}
%         \centering
         \includegraphics[width=0.24\linewidth]{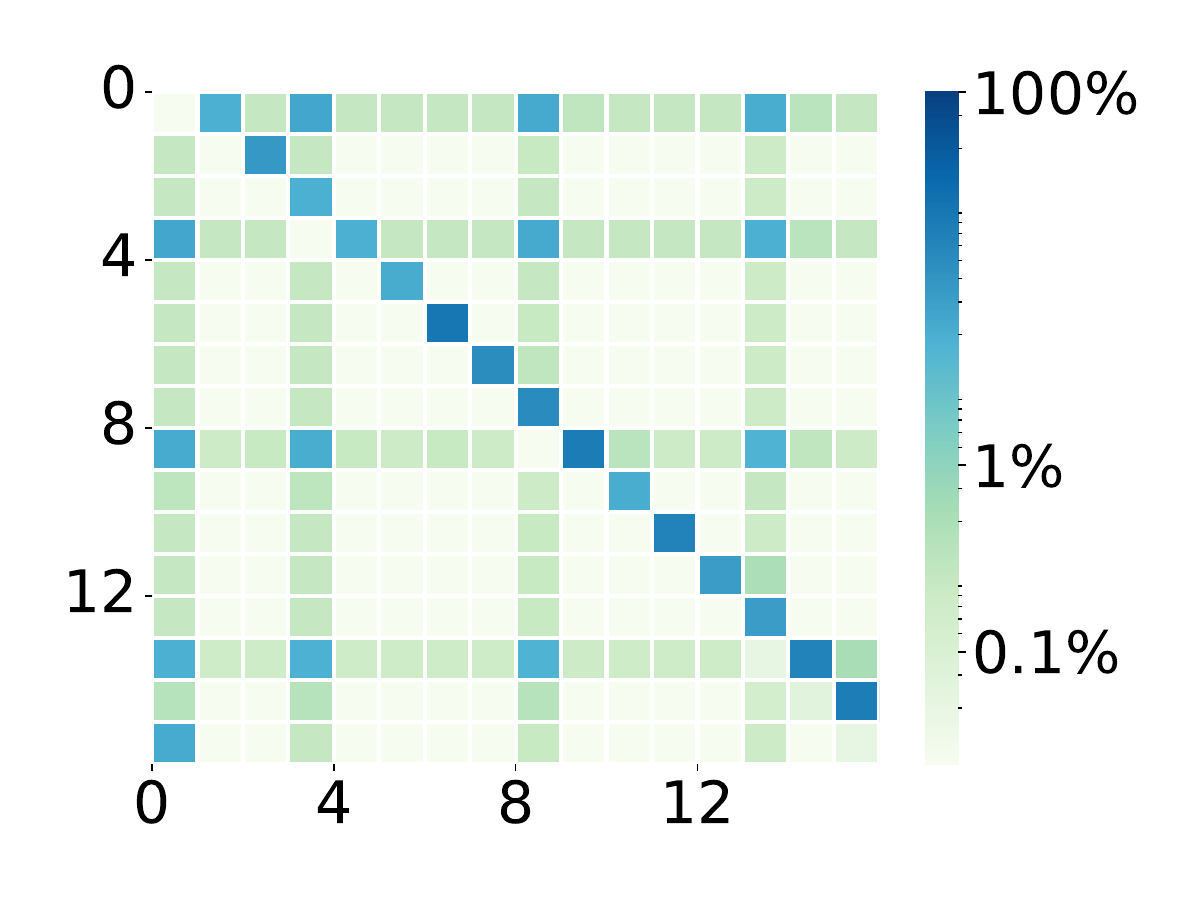}
%         \subcaption{DLRM hybrid parallelism}
%         % \vspace{-4mm}
%         \label{fig:residual-hybrid}
%     \end{subfigure} &
%     \begin{subfigure}{0.24\linewidth}
%         \centering
        \includegraphics[width=0.24\linewidth]{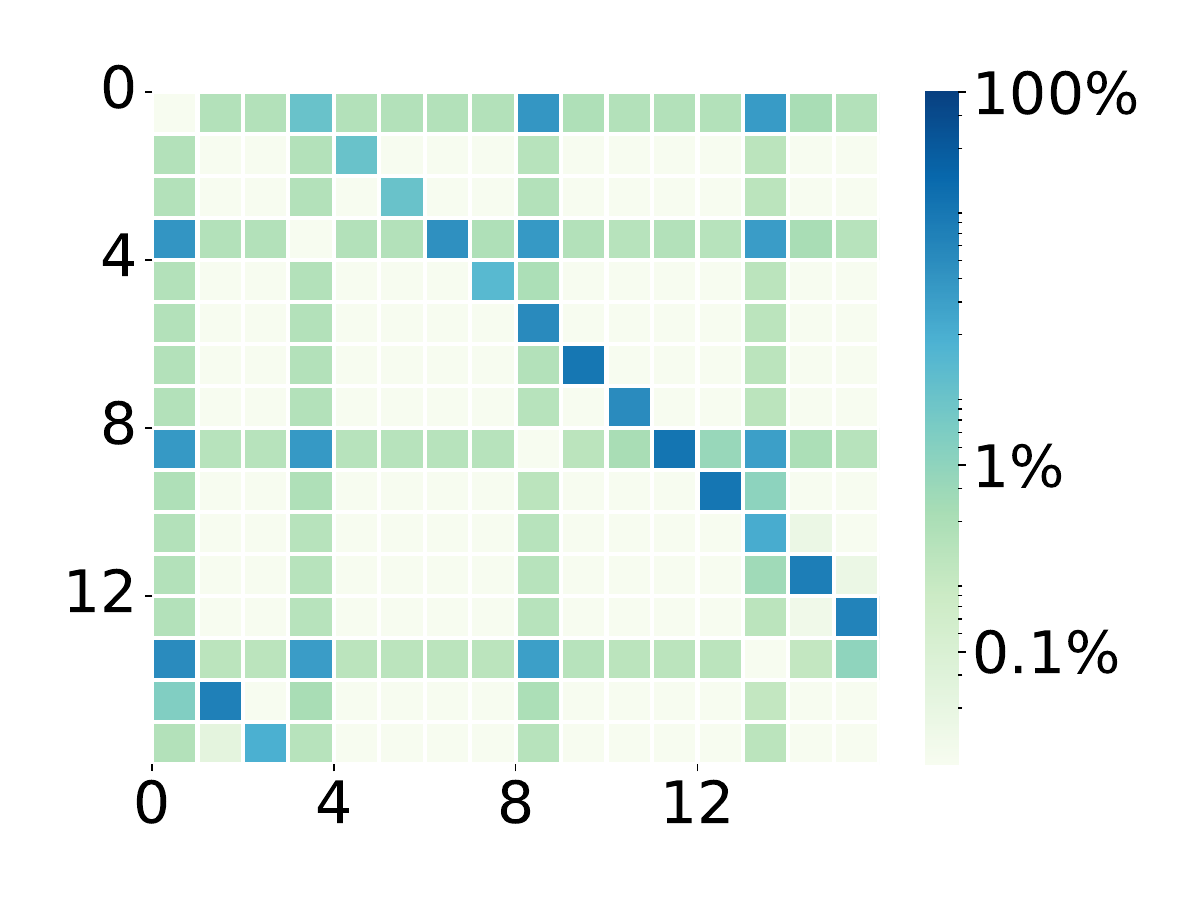}
%         \subcaption{DLRM $+3$ permuted}
%         % \vspace{-4mm}
%         \label{fig:residual-heatmap2}
%     \end{subfigure} &
%     \begin{subfigure}{0.24\linewidth}
%         \centering
        \includegraphics[width=0.24\linewidth]{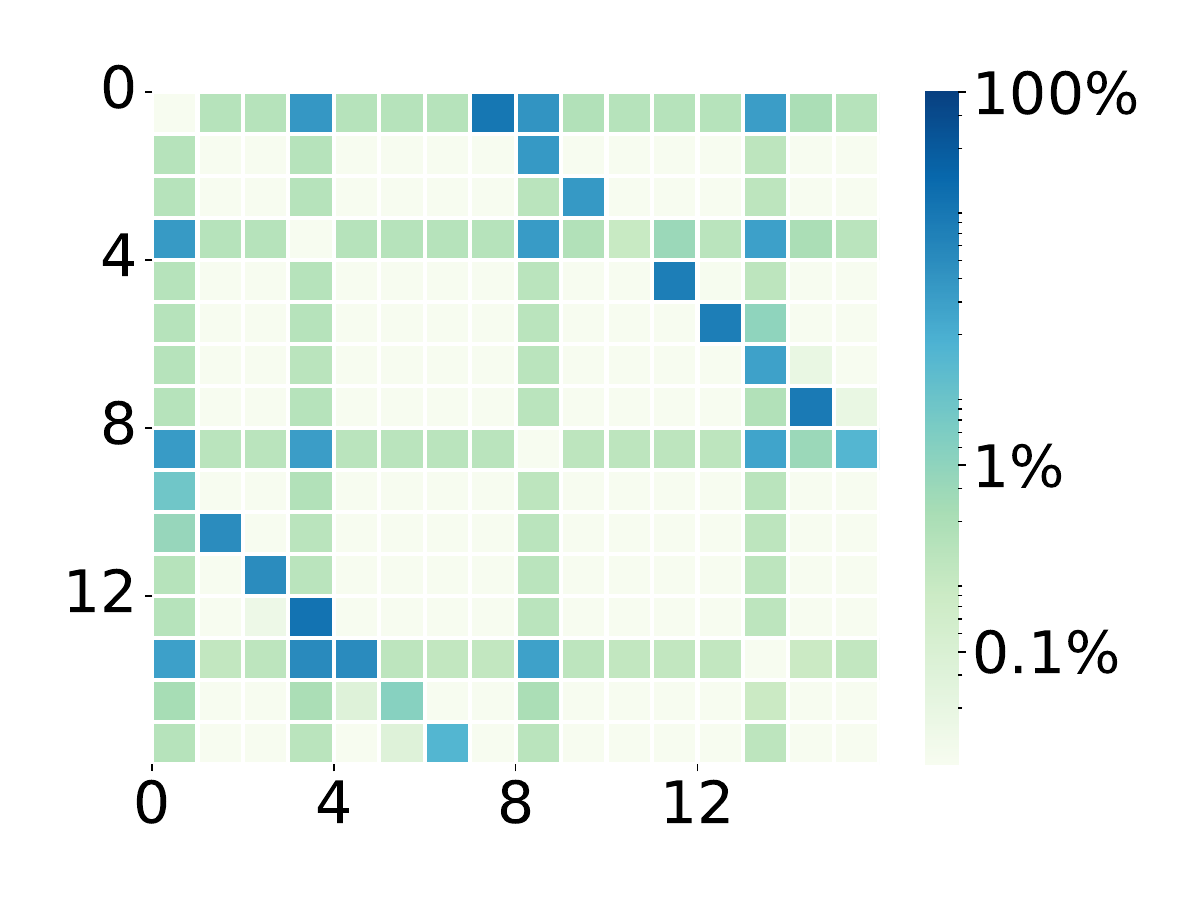} \\
%         \subcaption{DLRM $+7$ permuted}
%         % \vspace{-4mm}
%         \label{fig:residual-heatmap3}
%     \end{subfigure}
    \end{tabular}
    \noindent \textbf{(a) DLRM data parallelism} \ \ \textbf{(b) DLRM hybrid parallelism} \ \ \textbf{(c) DLRM $+3$ permuted} \ \ \textbf{(d) DLRM $+7$ permuted}
    \caption{The demand matrices of emerging Machine Learning workloads, particularly DNN training workloads, exhibit excellent structure: \first the decomposed floor and residual matrices are mostly regular (in terms of the sum of rows and columns); \second the floor matrices are typically close to a permutation matrix and carry majority of the traffic, typically $>75\%$ in each row and column; \third the residual matrices typically carry very low traffic, typically $<25\%$ in each row and column. The color of each entry (cell) in the heatmaps (demand matrices) indicates the demand specified by the entry as the minimum percentage of the corresponding total source (row) demand and the corresponding total destination (column) demand.}
    \label{fig:floor-residual-dm}
\end{figure*}

\subsection{Structure in ML Workloads}
\label{sec:workloads}

Since throughput is a function of the underlying demand matrix, we now briefly shift our focus on the emerging Machine Learning workloads in modern datacenters. 
The inherent structure in the demand matrices of such workloads enables us to study throughput on a set of demand matrices of special interest while improving the tractability in our analysis later. 

To better understand the structure of the demand matrices of machine learning workloads, we consider real-world measurements published recently~\cite{285119}. In particular, these matrices correspond to DNN training workloads over $16$ servers, each with $4\times 25$Gbps links (degree $4$). While the exact values of the demand matrices have not been made open source, we use simple image processing techniques to extract the values. As prior work already pointed out, a large portion of the demand matrix is well-structured due to collective communication \eg ring-AllReduce~\cite{doi:10.1177/1094342005051521}, and the remaining portion of the demand is more uniformly spread due to model parallelism~\cite{NIPS2014_7d6044e9}. We first normalize the extracted matrices to the link capacity $25$Gbps. We then decompose the matrices to \first floor, where each entry in the matrix is the largest integer less than the corresponding entry in the original matrix and \second residual, where each entry in the matrix is the difference between the corresponding entries in the original and the floor matrix. Figure~\ref{fig:floor-residual-dm} shows the resulting floor and residual matrices of four workloads presented in prior work~\cite{285119}.

Our observation is that the floor matrix $Int(\M)$ of a typical machine learning workload is mostly regular \ie the sum of every row and column falls within a small interval of values, in fact, Figure~\ref{fig:floor-residual-dm}~(a),~(b),~(c),~(d) (top row) show that the floor matrix is close to a permutation matrix with every row and column carrying $>75\%$ of the total demand in the corresponding row and column. The remaining portion of the demand $Res(\M)$ is spread across all the nodes as seen in Figure~\ref{fig:floor-residual-dm} (bottom row).

Following these observations, we focus on \emph{uniform-residual} demand matrices, formally defined in Definition~\ref{def:uniform-residual-demand}. Intuitively, a demand matrix is uniform-residual if the percentage of the total demand in every row and column of the corresponding residual matrix is within a small interval, in particular, within one of the three intervals: $0\%$-$25\%$, or $25\%$-$50\%$, or $50\%$-$100\%$. For instance, all the demand matrices presented in Figure~\ref{fig:floor-residual-dm} are uniform-residual since every row and column of the corresponding floor matrices carries at least $75\%$ of the corresponding row and column demand and the residual matrix carries at most $25\%$ of the demand.

% \begin{definition}[Uniformly skewed demand matrix]\label{def:skew-demand}
%     Given a demand matrix $\mathcal{M}=\{ \frac{m_{u,v}}{c} \mid u\in \mathcal{T}, v\in \mathcal{T} \}$ normalized to link capacity $c$, let $\mathcal{F}=\{ \lfloor\frac{m_{u,v}}{c} \rfloor \mid u\in \mathcal{T}, v\in \mathcal{T} \}$ be the corresponding floor matrix, where $\lfloor x \rfloor$ is the largest integer such that $\lfloor x \rfloor \le x$. The matrix $\mathcal{M}$ is uniformly skewed if the ratios of the sum of corresponding rows in $\mathcal{F}$ and $\mathcal{M}$ for all the rows \ie $\frac{\sum_{u\in \mathcal{T}} \lfloor m_{u,v} \rfloor }{\sum_{u\in \mathcal{T}} m_{u,v} }$, $\forall v\in\mathcal{T}$; and similarly for all the columns, are within the intervals: $[0, \frac{1}{2})$ or $[\frac{1}{2}, \frac{3}{4})$ or $[\frac{3}{4}, 1)$.
% \end{definition}

\begin{definition}[Uniform-residual demand matrix]\label{def:uniform-residual-demand}
A demand matrix $\M$ is uniform-residual if its normalized (by capacity $c$) matrix $\M^\prime = \frac{1}{c}\cdot \M$ has an integer-residual decomposition such that the ratios of the sum of each row $u$ in $Res(\M^\prime)$ and $\M^\prime$, \ie $\frac{\sum_{v} Res(\M^\prime)_{u,v}}{\sum_{v} \M^\prime_{u,v} }$, for all rows, and similarly for all the columns, fall within an interval: $[0, \frac{1}{4})$ or $[\frac{1}{4}, \frac{1}{2})$ or $[\frac{1}{2}, 1]$.
\end{definition}

Our definition of uniform-residual demand matrices not only captures the ML workloads shown in Figure~\ref{fig:floor-residual-dm},  but also captures several other matrices commonly used in the literature. For instance all-to-all uniform demand matrices or  permutation demand matrices are uniform-residual. 
% \chen{Type of Uniformly skewed
% \begin{itemize}
%     \item All two values matrices (e..g, permutation, uniform)
%     \item All matrices with $Res(M)=0_n$ (integer matrix)
%     \item All matrices that are "symmetric channel" (all rows and all columns are permutations of the same set). i.e., 
%     $\lfloor M \rfloor$ and $Res(M)$ are symmetric channel matrices.
%     \item All matrices with $\lfloor M \rfloor$ is scaled double stochastic (implying that $Res(M)$ is scaled double stochastic)
% \end{itemize}
% }

\medskip
\noindent{\textcolor{myred}{$\blacksquare$ \textbf{\textit{Takeaway.}}} \textit{The demand matrices of typical machine learning workloads are uniform-residual, meaning that the communication pattern of every source (and destination) is predominantly similar.}}

\subsection{Straightforward Throughput Bounds}
\label{sec:trivial}

We now present few straightforward bounds on the throughput of reconfigurable networks, particularly for uniform-residual demand matrices. We focus on demand-aware static and demand-aware reconfigurable networks (see \S\ref{sec:rdcns}) and prove two straight-forward results on their throughput.

The throughput of a topology is largely affected by the number of hops required to transit each source-destination demand. To this end, if the demand matrix is such that it allows adding direct links to corresponding source-destination demands, then the topology can achieve full throughput. To this end, both demand-aware static and demand-aware reconfigurable networks can achieve full throughput by a one-shot reconfiguration for all matrices that are uniform-residual and the normalized matrix (normalized to capacity) consists of only integer values.

\begin{theorem}[Throughput under integer demand matrices]\label{th:throughput-integer}
The throughput of a demand-aware reconfigurable network is~$1$ (full-throughput), specifically for those demand matrices for which the normalized demand matrix (normalized by link capacity) is equal to the corresponding floor matrix. 
\end{theorem}

\begin{proof}
Given that the normalized demand matrix is equivalent to the corresponding floor matrix, all entries in the normalized demand matrix are integers. To achieve full throughput, a one-shot reconfiguration of the circuit switching, with one matching at each switch, suffices. Specifically, for an entry with value $x$ in the normalized demand matrix, we add $x$ number of links between the corresponding source and destination via circuit switches. Since the total demand originating from and destined to each node cannot exceed its capacity, there are always a sufficient number of links available to satisfy the demand over single-hop paths.
\end{proof}

While Theorem~\ref{th:throughput-integer} specifically applies to certain types of demand matrices, its validity extends to any scenario with a reasonable reconfigurable delay, given that a one-shot reconfiguration suffices. We now shift our focus to encompass all demand matrices within the hose model, under the assumption of negligible reconfiguration delay \ie any number of reconfigurations can be performed over time without any overhead. The throughput of reconfigurable networks under such an assumption has been implicitly known to be $1$ in prior works~\cite{10.1145/2716281.2836126,10.1145/2486001.2486007}; we state it here with a proof to formally establish a throughput upper bound.

\begin{theorem}[Ideal throughput of demand-aware RDCN]\label{th:throughput-ideal}
The throughput of a demand-aware reconfigurable network is $1$ \ie full-throughput for any demand matrix if the reconfiguration delay is negligible. 
\end{theorem}

\begin{proof}
Within the hose model set of demand matrices, we consider saturated demand matrices \ie the sum of every row (column) equals the outgoing (incoming) capacity of each node. If a topology can achieve throughput $\theta$ for all saturated demand matrices, then the topology can achieve throughput $\theta$ for any demand matrix~\cite{10.1145/3452296.3472913}. Given that saturated demand matrices are doubly stochastic, we first decompose the matrix using Birkhoff–von Neumann (BvN) decomposition technique~\cite{birkhoff1946three} into $k$ permutation matrices, where $k$ can be up to $n^2$. Let $\mathcal{M}$ be any saturated demand matrix, where the sum of every row and column is $c\cdot u$ (total capacity of each node). Let the corresponding BvN decomposition be $\mathcal{M} = \lambda_1 \cdot P_1 + \lambda_2\cdot P_2 ... + \lambda_k P_k$, where $P_i$ is a permutation matrix and the coefficients $\lambda$ are such that $\sum_{i=1}^k \lambda = c\cdot u$.
Using this decomposition, we configure the topology such that each permutation $P_i$ is executed using full node capacity $c\cdot u$ for $\frac{\lambda_i}{c\cdot u}\cdot \Delta$ units of time over a period of one unit of time $\Delta$. Over $\Delta$ amount of time, $\lambda_i\cdot P_i$ portion of the demand matrix generates $\lambda_i\cdot P_i\cdot \Delta$ demand in volume. As a result, during $\frac{\lambda_i}{c\cdot u}\cdot \Delta$ amount of time, by executing the corresponding permutation $P_i$ using full capacity $c\cdot u$, the topology can fully satisfy $\lambda_i\cdot P_i$ portion of the demand matrix.
As a result, the topology can fully satisfy the demand matrix $\mathcal{M}$ over each period of one unit of time $\Delta$ and achieves full throughput. 
\end{proof}

\medskip
\noindent{\textcolor{myred}{$\blacksquare$ \textbf{\textit{Takeaway.}}}
Demand-aware reconfigurable networks can ideally achieve full-throughput if the reconfiguration delays are negligible $\approx 0$. However, the achievable throughput under realistic reconfiguration delays is still unclear. 

\subsection{RoadMap}
\label{sec:roadmap}

Our analysis in \S\ref{sec:trivial} gives an upper bound of $1$ for the throughput of demand-aware reconfigurable networks, illustrating the potential throughput that such networks could achieve if reconfiguration delays can be reduced arbitrarily close to zero.
In contrast, the throughput of demand-oblivious reconfigurable networks is still bound by $\frac{1}{2}$ even within the set of uniform-residual demand matrices since permutation matrices (the worst-case~\cite{10.1145/3579312}) are a subset of uniform-residual matrices. 

To better understand the \emph{achievable} throughput (lower bound) of demand-aware reconfigurable networks under realistic reconfiguration delays, we need to first answer the following fundamental questions:

\begin{itemize}[leftmargin=*,label=\small{\textcolor{myred}{$\blacksquare$}}]
    \item What is the throughput achievable by a demand-aware static network \ie a demand-aware network where a one-shot reconfiguration is allowed once in a large interval of time? (\S\ref{sec:demand-aware-static})
    \item What is the throughput achievable by a demand-aware reconfigurable network if the circuit switching schedule is restricted (simplified) to be periodic and of fixed-duration?~(\S\ref{sec:demand-aware-periodic})
\end{itemize}

By answering the above questions, we directly establish a lower bound for the throughput of demand-aware reconfigurable networks. This is because both demand-aware static and demand-aware periodic networks fall within the broader category of general demand-aware reconfigurable networks.

\section{Throughput Landscape of RDCNs}
\label{sec:throughput}
Building upon on our observations in \S\ref{sec:motivation}, in this section, we primarily focus on the throughput of demand-aware static (\S\ref{sec:demand-aware-static}) and demand-aware periodic networks (\S\ref{sec:demand-aware-periodic}). We establish a throughput lower bound of $\frac{2}{3}$ for both these types of networks. This finding implies a general lower bound of $\frac{2}{3}$ for the throughput of demand-aware reconfigurable networks as a whole. Interestingly, the technique to decompose a demand matrix into floor and residual matrices plays a crucial role in our throughput analysis of demand-aware networks, as we will see later in this section.
Our introduction of demand-aware periodic networks represents a novel contribution, and our throughput analysis draws on an interesting connection between demand-aware static and demand-aware periodic networks (\S\ref{sec:demand-aware-periodic}). 

\subsection{Demand-aware Static Networks}\label{sec:demand-aware-static}

We consider demand-aware static networks (described in \S\ref{sec:rdcns}) with $n$ ToR switches, each equipped with $u$ incoming and outgoing links, $u$ optical circuit switches, each having $n$ input and output ports. Our analysis is confined to the scenario where $u=n$, which will later be crucial for our analysis of demand-aware periodic networks.

In order to prove a throughput lower bound of $\frac{2}{3}$ under uniform-residual demand matrices, it is sufficient to consider the matrices for which the sum of every row (source) and every column (destination) equals $\frac{2}{3}$ fraction of the total node capacity $c\cdot n$ \ie doubly-stochastic matrices\footnote{If the matrix is not doubly stochastic, within the upper limit of $\frac{2\cdot c\cdot n}{3}$ for each row and column, then the matrix can be augmented by a non-negative valued demand matrix to convert it to doubly stochastic. Since we augment by a non-negative demand matrix, the throughput of the original demand matrix cannot be smaller than that of the doubly stochastic matrix.}, and showing that every source-destination demand can be satisfied in the network. We call such matrices $\frac{2}{3}$-saturated demand matrices.

Across the set of all $\frac{2}{3}$-saturated demand matrices that are uniform-residual, we prove our lower bound in three steps based on the range of percentage of total demand in every row and column of the corresponding floor matrices \first between $0\%$-$50\%$ (Lemma~\ref{lem:lb-interval1}) \second between $50\%$-$75\%$ (Lemma~\ref{lem:lb-interval2}) and \third between $75\%$-$100\%$ (Lemma~\ref{lem:lb-interval3}). Note that a demand-aware static network only executes one matching in each circuit switch. As a result, in the following, we first analyze ToR-to-ToR graphs of degree $n$ which we later decompose to $n$ matchings corresponding to each optical switch.

\begin{Lemma}\label{lem:lb-interval1}
Given any $\frac{2}{3}$-saturated demand matrix $\mathcal{M}$ that is uniform-residual (Definition~\ref{def:uniform-residual-demand}) within the interval $[\frac{1}{2}, 1]$, then a demand-aware static network of degree $n$ can fully satisfy the demand.   
\end{Lemma}

\begin{proof}
Consider a ToR-to-ToR graph of degree $n$ that forms a complete graph. 
Note that each row and column in the residual matrix accounts for more than $\frac{1}{2}$ 
of the corresponding total row and column demand. 
Further, by definition, each entry in the residual matrix is strictly less than $1$. 
Consequently, in a complete graph where there is one link between every source-destination pair, at least $\frac{1}{2}$ of the demand corresponding 
to the residual matrix can be transmitted on a single-hop. 
This translates to at least $\geq \frac{1}{2} \cdot \frac{2 \cdot c \cdot n}{3} 
\geq \frac{c \cdot n}{3}$ of the demand from each source and towards 
each destination being transmitted on a single-hop.
Moreover, a load-balancing scheme can be devised such that even if the demand 
from the floor matrix is transmitted over paths of length 2, the total incoming 
and outgoing capacity utilized by each node will be at most $(\frac{c \cdot n}{3} + x) 
+ (\frac{c \cdot n}{3} - x) \cdot 2$, which is less than or equal to the total capacity $c \cdot n$. 
Here, $\frac{c \cdot n}{3} + x$ denotes the exact amount of demand 
in the residual matrix, $\frac{c \cdot n}{3} - x$ denotes the demand 
from the floor matrix (since the total is $\frac{2 \cdot c \cdot n}{3}$), 
and $x$ is greater than zero because every row and column 
in the floor matrix carries a fraction of the total demand that is less 
than $\frac{1}{2}$, i.e., less than $\frac{c \cdot n}{3}$ in total. 
This proves that a complete graph can support any demand matrix specified 
in Lemma~\ref{lem:lb-interval1}. Decomposing the complete graph into $n$ matchings, 
and executing one matching at each of the $n$ optical switches allows 
the demand-aware static network to emulate a complete graph, achieving 
a throughput of $\frac{2}{3}$.
\end{proof}

\begin{Lemma}\label{lem:lb-interval2}
Given any $\frac{2}{3}$-saturated demand matrix $\mathcal{M}$ that is uniform-residual (Definition~\ref{def:uniform-residual-demand}) within the interval $[\frac{1}{4}, \frac{1}{2})$, then a demand-aware static network of degree $n$ can fully satisfy the demand. 
\end{Lemma}

\begin{proof}
We begin by decomposing the matrix $\mathcal{M}$ into floor and residual matrices. Note that each row and column in the floor matrix accounts for at least $\frac{1}{2}$ (residual is at most $\frac{1}{2}$) and at most $\frac{3}{4}$ (residual is at least $\frac{1}{4}$) of the corresponding total row and column demand. For every entry in the floor matrix with a value of value $x$, we add $x$ number of links between the corresponding source and destination. As a result, the entire demand represented by the floor matrix can be transmitted over single-hop. This approach ensures that at least $\frac{1}{2} \cdot \frac{2 \cdot c \cdot n}{3} \ge \frac{c \cdot n}{3}$ of the demand from every source and towards every destination is satisfied; it also utilizes at most $\frac{n}{2}$ links from each node, given that every row and column in the floor matrix sums up to at most $\frac{3}{4} \cdot \frac{2}{3} = \frac{1}{2}$. We are now left with the residual matrix and at least $\frac{n}{2}$ links at each node. We construct a random regular graph of degree $\frac{n}{2}$. A link between a source-destination pair fully satisfies the portion of the demand specified by the corresponding entry in the residual matrix (since every entry is strictly less than $1$). As a result, the random regular graph satisfies at least $\frac{1}{2}$ of the demand in the residual matrix on single-hop \ie $\ge \frac{1}{2}\cdot\frac{c\cdot n}{3}\ge \frac{c\cdot n}{6}$ demand from each row and from each column. The rest of the demand ($\frac{c\cdot n}{6}$) can be transmitted in $2$-hop paths, essentially consuming at most $(\frac{c\cdot n}{6})+(\frac{c\cdot n}{6})\cdot 2 = \frac{c\cdot n}{2}$ from each node which is within the budget of $\frac{n}{2}$ links to satisfy the residual matrix. Overall, both the floor and residual matrices can be transmitted within the capacity limits of the network.
\end{proof}

\begin{Lemma}\label{lem:lb-interval3}
Given any $\frac{2}{3}$-saturated demand matrix $\mathcal{M}$ that is uniform-residual (Definition~\ref{def:uniform-residual-demand}) within the interval $[0, \frac{1}{4})$, then a demand-aware static network of degree $n$ can fully satisfy the demand. 
\end{Lemma}

\begin{proof}
Our proof follows a methodology similar to that of Lemma~\ref{lem:lb-interval2}. We begin by decomposing the matrix $\mathcal{M}$ into floor and residual matrices. Note that each row and column in the floor matrix accounts for at least $\frac{3}{4}$ (residual is at most $\frac{1}{4}$) and at most the entire portion (residual is at least $0$) of the corresponding total row and column demand. For every entry in the floor matrix with a value of value $x$, we add $x$ number of links between the corresponding source and destination. As a result, the entire demand represented by the floor matrix can be transmitted over single-hop. This approach ensures that at least $\frac{3}{4} \cdot \frac{2 \cdot c \cdot n}{3} \ge \frac{c \cdot n}{2}$ of the demand from every source and towards every destination is satisfied. Additionally, it utilizes at most $\frac{2\cdot n}{3}$ links from each node, given that every row and column in the floor matrix sums up to at most $1 \cdot \frac{2}{3} = \frac{2}{3}$. We are now left with the residual matrix and at least $\frac{n}{3}$ links at each node. Note that each row and column in the residual matrix accounts for at most $\frac{1}{4}$ portion of the corresponding row and column demand \ie at most $\frac{1}{4}\cdot\frac{2\cdot c\cdot n}{3}= \frac{c\cdot n}{6}$ demand. By constructing a regular graph of degree $\frac{n}{3}$, the residual demand can be transmitted within capacity limits even if all the demand is transmitted on $2$-hop indirect paths \ie $\frac{c\cdot n}{6}\cdot 2 \le \frac{n}{3}$. 
\end{proof}

The throughput lower bound follows from our results in Lemma~\ref{lem:lb-interval1},~\ref{lem:lb-interval2},~\ref{lem:lb-interval3}.

\begin{theorem}[Lower bound for demand-aware static RDCNs]\label{th:lb-static}
The throughput of a demand-aware static network with $n$ ToRs each with $u=n$ incoming and outgoing links, is lower bounded by $\frac{2}{3}$ under uniform-residual demand matrices.
\end{theorem}

So far, our analysis suggests that demand-aware static networks can achieve at least a throughput of $\frac{2}{3}$ for uniform-residual demand matrices. We next focus on the upper bound for the throughput of such networks. In order to prove an upper bound $\theta^*$ on the throughput, it is sufficient to show that there exists a demand matrix such that the network cannot support more than $\theta^*$ throughput. 

\begin{theorem}[Upper bound]\label{th:ub-static}
The throughput of demand-aware static networks with $n$ ToRs each with $u=n$ incoming and outgoing links, is upper bounded by $\frac{4}{5}$.
\end{theorem}

\begin{proof}
We prove our claim using a demand matrix $\mathcal{M}$ that specifies $0.5$ and $1.5$ demand (normalized by capacity $c$) alternatively in every row and column as follows:
\begin{align*}
 \mathcal{M}=&
    \begin{bmatrix}
        0.5 & 1.5 & 0.5 &... & 1.5\\
        1.5 & 0.5 & 1.5 &... & 0.5\\
        \vdots & \vdots & \vdots & \vdots & \vdots\\
        0.5 & 1.5 & 0.5 &... & 1.5\\
        1.5 & 0.5 & 1.5 &... & 0.5\\
    \end{bmatrix}
    % \\
    % =&
    % \underbrace{
    % \begin{bmatrix}
    %     0 & 1 & 0 &... & 1\\
    %     1 & 0 & 1 &... & 0\\
    %     ... & ... & ... &... & ...\\
    %     0 & 1 & 0 &... & 1\\
    %     1 & 0 & 1 &... & 0\\
    % \end{bmatrix}
    % }_{Floor\ matrix}
    % +
    % \underbrace{
    % \begin{bmatrix}
    %     0.5 & 0.5 & 0.5 &... & 0.5\\
    %     0.5 & 0.5 & 0.5 &... & 0.5\\
    %     ... & ... & ... &... & ...\\
    %     0.5 & 0.5 & 0.5 &... & 0.5\\
    %     0.5 & 0.5 & 0.5 &... & 0.5\\
    % \end{bmatrix}
    % }_{Residual\ matrix}
\end{align*}
The above demand matrix is saturated \ie the sum of every row and column equals $n$, the total number of incoming and outgoing links at each node. By greedily adding links between source-destination pairs, at most $\frac{3 \cdot n}{4}$ of the total demand from each row and column can be satisfied in a single-hop. This results in at least $\frac{n}{4}$ of the demand from every row and column requiring transmission over at least 2-hops. As a result, the above demand matrix consumes at least $(\frac{3 \cdot n}{4}) \cdot 1 + (\frac{n}{4}) \cdot 2 = \frac{5 \cdot n}{4}$ of the total capacity for each node, while the total capacity available is only $n$. Therefore, the maximum scaling factor required for the demand matrix to be feasible within the capacity limits is at most $\frac{4}{5}$.
\end{proof}

\subsection{Demand-aware Periodic Networks}\label{sec:demand-aware-periodic}

We now introduce demand-aware \emph{periodic} networks based on \emph{fixed-duration} reconfigurations. These networks are similar to demand-oblivious reconfigurable networks such as RotorNet~\cite{10.1145/3098822.3098838} and Sirius~\cite{10.1145/3387514.3406221} but the periodic circuit switching schedule is derived based on the demand matrix\footnote{Circuit switches in RotorNet (tunable lasers in Sirius) execute a schedule that is installed at initialization time and cannot be changed (or configured) at run-time, irrespective of the underlying demand matrix.}. For instance, in an architecture like RotorNet, we assume that the control plane measures the demand matrix and computes a \emph{periodic} switching schedule, where each matching in the schedule is executed by the optical circuit switches for a fixed-duration of time and each switch takes a fixed amount of time to reconfigure to the next matching in the schedule.
This can be extended to an architecture like Sirius by interpreting the switching schedule computed by the control plane as the schedule for tuning the lasers such that the same set of matchings are achieved. Demand-aware periodic networks are particularly attractive due to their simplicity and the capability of their circuit switches to reconfigure at nanosecond timescales. These networks are also practically realizable, provided they incorporate additional functionality that enables the updating of switching schedules in rotor switches (or tunable lasers) during run-time. 

Understanding the throughput of demand-aware periodic networks directly establishes a lower bound for the throughput of demand-aware reconfigurable networks as a whole. Our throughput analysis of demand-aware periodic networks relies on an important result from prior work that states: the throughput of a periodic network is equivalent to the throughput of a static graph obtained from the union of graphs (of the periodic network) over one period of time~\cite{10.1145/3579312}. Specifically, let $\mathcal{G}_t=(V, \mathcal{E}_t)$ denote the ToR-to-ToR graph at timeslot $t$ of a periodic network and let $\Gamma$ denote the period of the circuit switching schedule. The circuit switches implement a matching for $\Delta$ duration of time (one timeslot) and it takes $\Delta_r$ amount of time to reconfigure to the next matching. Prior work~\cite{10.1145/3579312} proves that the throughput of a periodic network represented by $\mathcal{G}_t$ is equivalent to a static graph $G=(V, E)$, where $E=\cup_{t\in[0,\Gamma)} \mathcal{E}_t$. The capacity of each link in the static graph is $\frac{c}{\Gamma}$, where $c$ is the capacity of each link in the original periodic graph. As a consequence of this result, we obtain the following Corollary that establishes the relation between the throughput of demand-aware static and demand-aware periodic networks. 

\begin{corollary}\label{cor:static-periodic}
The throughput of a demand-aware static network with $n$ ToRs, each with $n$ incoming and outgoing links with capacity $\frac{c}{\Gamma}$ is equivalent to the throughput of a demand-aware periodic network with $n$ ToRs, each with $u\le n$ incoming and outgoing link with capacity $c$, where $\Gamma=\frac{n}{u}$ is the period of the periodic schedule.
\end{corollary}

\begin{proof}
The demand-aware static network outlined in Corollary~\ref{cor:static-periodic} can be represented as a ToR-to-ToR directed graph of degree $n$ with link capacities $\frac{c}{\Gamma} = \frac{c \cdot u}{n}$. Considering a demand-aware periodic network with $n$ ToRs, each having $u \le n$ incoming (and outgoing) links of capacity $c$, we utilize the aforementioned static graph to derive the switching schedule for the periodic network. Since the static graph is regular, meaning that each ToR has $n$ incoming and outgoing links, it can be decomposed into $n$ perfect matchings. By shuffling these matchings and installing $\frac{n}{u}$ of them in each of the $u$ optical circuit switches, we ensure that their union reconstructs the original static graph.
As a result, the periodic network effectively \emph{emulates} a static graph identical to the ToR-to-ToR graph of the demand-aware static network with link capacity $\frac{c}{\Gamma}$. Here, $\Gamma$, the period of the switching schedule, is $\frac{n}{u}$, as each switch sequentially executes $\frac{n}{u}$ matchings in a periodic manner.
The proof follows since the throughput of a periodic graph and its corresponding static emulated graph are equivalent~\cite{10.1145/3579312}.
\end{proof}

Although our analysis of demand-aware static networks (\S\ref{sec:demand-aware-static}) is confined to scenarios with a large number of incoming and outgoing links at each ToR (degree $= n$), it is relevant to the throughput of demand-aware periodic networks for \emph{any} degree $u\le n$. Based on Corollary~\ref{cor:static-periodic}, we can now analyze the throughput of demand-aware periodic networks of any degree. We obtain the following Corollaries on the throughput of demand-aware periodic networks, as a direct consequence of our results in Theorem~\ref{th:lb-static} and Corollary~\ref{cor:static-periodic}. Our formal proof appears in Appendix~\ref{sec:deferred-proofs}.

\begin{corollary}[Lower bound]\label{cor:throughput-periodic-lb}
The throughput of demand-aware periodic networks is lower bounded by $\frac{2}{3}$ under uniform-residual demand matrices (Definition~\ref{def:uniform-residual-demand}).
\end{corollary}

% \begin{proof}
% Since the throughput of a demand-aware static network of degree $n$ with link capacity $\frac{c\cdot u}{n}$ is lower bounded by $\frac{2}{3}$ for uniform-residual demand matrices, proof follows from Corollary~\ref{cor:static-periodic} for the throughput of a demand-aware periodic network with degree $u$ and link capacity $c$.
% \end{proof}

\begin{corollary}[Upper bound]\label{cor:throughput-periodic-ub}
The throughput of demand-aware periodic networks is upper bounded by $\frac{4}{5}$.
\end{corollary}

% \begin{proof}
% Since the throughput of a demand-aware static network of degree $n$ with link capacity $\frac{c\cdot u}{n}$ is at most $\frac{4}{5}$ from Theorem~\ref{th:ub-static}, the throughput of a demand-aware periodic network cannot be greater than $\frac{4}{5}$ due to Corollary~\ref{cor:static-periodic}.
% \end{proof}

\subsection{Throughput Landscape in Summary}
We now present a summary of the throughput landscape for reconfigurable datacenter networks, contextualizing prior research alongside the key results of this paper. Figure~\ref{fig:intro} illustrates the throughput landscape.

\noindent \textbf{Demand-oblivious reconfigurable (Prior work):} The throughput of demand-oblivious networks is tightly bounded by $\frac{1}{2}\cdot (1-\Delta_r)$~\cite{10.1145/3579312,10.1145/3491050,10.1145/3519935.3520020}, where $\Delta_r$ is the fraction of time spent in reconfigurations (typically $\Delta_r>0.9$~\cite{10.1145/3098822.3098838,10.1145/3387514.3406221}). The worst-case throughput of $\frac{1}{2}\cdot (1-\Delta_r)$ is achieved under permutation demand matrices~\cite{10.1145/3579312}. These networks correspond to oblivious and \emph{fixed-duration} reconfigurations in Figure~\ref{fig:intro}.

\noindent \textbf{Demand-aware static (This work):} The throughput of demand-aware static networks is lower bounded by $\frac{2}{3}$ (Theorem~\ref{th:lb-static}) and upper bounded by $\frac{4}{5}$ (Theorem~\ref{th:ub-static}) for uniform-residual demand matrices. The upper bound of $\frac{4}{5}$ is achieved under a demand matrix that specifies alternating $0.5$ and $1.5$ (normalized by capacity) between source-destination pairs in the network.

\noindent \textbf{Demand-aware periodic (This work):} Based on Corollary~\ref{cor:static-periodic}, all our results on demand-aware static networks, transfer to demand-aware periodic networks as well \ie the throughput is lower bounded by $\frac{2}{3}$ and upper bounded by $\frac{4}{5}$ for uniform-residual demand matrices. These networks correspond to demand-aware and \emph{fixed-duration} reconfigurations in Figure~\ref{fig:intro}.

\noindent \textbf{Demand-aware reconfigurable (This work):} Since the throughput of demand-aware periodic networks is lower bounded by $\frac{2}{3}$, the throughput of demand-aware reconfigurable networks as a whole is also lower bounded by $\frac{2}{3}$. The throughput upper bound is $1$ (Theorem~\ref{th:throughput-ideal}). These networks correspond to demand-aware and \emph{variable-duration} reconfigurations in Figure~\ref{fig:intro}.

\noindent \textbf{Separation of demand-aware \& demand-oblivious:} In this work, our results demonstrate a distinct separation in throughput between demand-aware and demand-oblivious networks under uniform-residual demand matrices. Specifically, demand-aware networks can achieve a throughput that is at least $16\%$ higher than that of demand-oblivious networks in the worst-case scenario.

\begin{figure*}
    \center
    % \begin{subfigure}{0.24\linewidth}
        % \centering
        \includegraphics[trim={0.65cm 1.35cm 0 1.3cm},clip,width=0.9\linewidth]{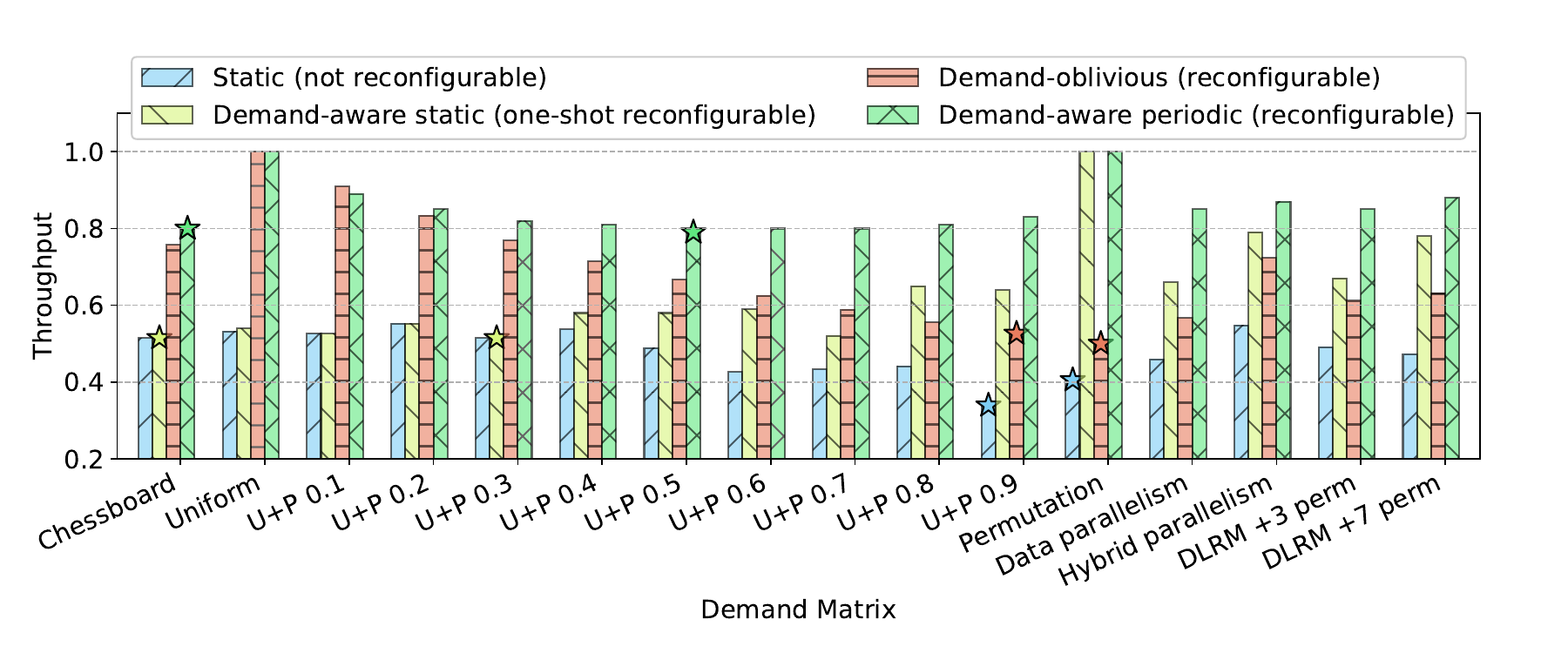}
        % \subcaption{Floor of DLRM data parallelism}
        % \vspace{-4mm}
    \caption{The throughput of demand-aware periodic networks is strictly superior to demand-oblivious and static networks across all demand matrices. Demand-aware static performs poorly under synthetic demand matrices due to low degree, but it outperforms demand-oblivious for DNN training demand matrices (last four on the right). Worst-cases for each network are indicated by $\star$.}
    \label{fig:degree-4}
    \vspace{-2mm}
\end{figure*}

\section{Empirical Evidence}\label{sec:eval}

In this section, we validate the throughput bounds established in \S\ref{sec:throughput}. Additionally, our empirical evaluation aims at answering the following questions.

\medskip
\noindent \textbf{(Q1)} \textit{In contrast to worst-case throughput, how do demand-oblivious and demand-aware networks fare against each other in terms of throughput for a given demand matrix?}

We find that demand-aware periodic networks 
out-perform alternative network designs (by up to $30\%$) across all the demand matrices considered in our evaluation, confirming their superior throughput bounds.

\medskip
\noindent \textbf{(Q2)} \textit{How does degree (incoming \& outgoing links) of the physical topology affect the throughput of demand-aware networks?}

Our evaluation shows an interesting relation between throughput and degree for demand-aware static networks. In particular, throughput increases with degree, reaching the throughput of demand-aware periodic networks for large degree. Our evaluation confirms our results in \S\ref{sec:demand-aware-periodic}, showing that the throughput of demand-aware periodic networks is largely independent of degree. 

\subsection{Methodology}
Our evaluation is based on a linear programming approach, using Gurobi~\cite{gurobi}.

\myitem{Network:} We consider a network consisting of $n=16$ ToRs. We initially set the number of incoming and outgoing links (degree) to $u=4$ and later vary between $4$ to $16$. We assume that each link is of capacity $25$Gbps. This network corresponds to the testbed used in prior work~\cite{285119}.

\myitem{Demand matrices:} We evaluate each network using the four demand matrices corresponding to DNN training~\cite{285119} \ie data parallelism, hybrid parallelism, DLRM $+3$, $+7$ permuted. Additionally, we consider the following synthetic demand matrices: \first Chessboard, the demand matrix used in Theorem~\ref{th:ub-static}, \second Uniform, the best-case demand matrix for demand-oblivious, \third Permutation, the worst-case demand matrix for demand-oblivious \fourth a combination of uniform and permutation \ie $\alpha$ times permutation $P$ and $(1-\alpha)$ times uniform $U$, denoted by $U+P$ $\alpha$. We consider $\alpha$ values between $0.1$ to $0.9$.
In total, we consider $16$ demand matrices.

\myitem{Comparisons:} We compare four networks \first \textbf{\textcolor{myblue}{static}} (no reconfigurations \eg expander), \second \textbf{\textcolor{myyellow}{demand-aware static}} (one-shot reconfiguration), \third \textbf{\textcolor{mydarkred}{demand-oblivious}} (\eg RotorNet~\cite{10.1145/3098822.3098838}), and \fourth \textbf{\textcolor{mygreen}{demand-aware periodic}} (see \S\ref{sec:demand-aware-periodic}).

\myitem{Computing throughput:} We use a combination of linear programming and heuristics to compute throughput for each type of network. Our linear program is based on standard multi-commodity flow formulation, with the objective to maximize throughput (Appendix~\ref{sec:linear-program}). As a result, routing and congestion control are optimal for each type of network, and the obtained throughput values correspond to the \textit{ideally} achievable throughput. We construct a static network using random regular graphs. We use complete graph for demand-oblivious networks due to their throughput equivalence~\cite{10.1145/3579312}. Computing optimal topologies with throughput maximization objective turns out to be impractical even for a $16$ ToR network. In fact, prior work resorted to maximum link utilization as an objective~\cite{10.1145/3544216.3544265}.
Leveraging our integer-residual decomposition technique outlined in our proofs in \S\ref{sec:throughput}, we adopt an iterative approach in steps of $0.01$ (resulting in an error margin of $0.01$) to find the maximum throughput for demand-aware networks (Appendix~\ref{sec:linear-program}).

\vspace{-3mm}
\subsection{Results}
\myitem{Demand-aware periodic outperforms in throughput: }
As evidenced by our worst-case bounds in \S\ref{sec:demand-aware-periodic}, our results in Figure~\ref{fig:degree-4} show that the demand-aware periodic network outperforms for every demand matrix. The lowest throughput across all demand matrices is $0.8$, achieved under chessboard demand matrix. Interestingly, chessboard matrix in Corollary~\ref{cor:throughput-periodic-ub} gives an \emph{upper bound} of $0.8$, hinting that the lower bound of $\frac{2}{3}$ in Corollary~\ref{cor:throughput-periodic-lb} can potentially be improved in the future. Starting from uniform demand matrix, as $\alpha$ increases, the throughput drops close to $0.8$ between $\alpha=0.5$ and $\alpha=0.7$ but improves beyond $\alpha=0.7$ and reaches throughput $1$ for permutation demand matrix. Interestingly, permutation demand matrix is the worst-case for demand-oblivious networks, achieving a throughput of only $0.5$. Even for the DNN training workloads, demand-aware periodic networks achieve a high throughput. In particular, demand-aware networks improve throughput by $49.9\%$ for data parallelism, by $20.1\%$ for hybrid parallelism, by $38.9\%$ for DLRM $+3$ permuted, and by $39.2\%$ for DLRM $+7$ permuted, compared to demand-oblivious networks.

\noindent \textbf{Worst-case throughput of demand-aware periodic is independent of degree:} Both lower and upper bound for the throughput of demand-aware periodic established in \S\ref{sec:demand-aware-periodic} are independent of degree. This can be confirmed based on our results in Figure~\ref{fig:worst-case}, showing that the worst-case throughput across all our matrices remains similar for degrees between $4$ to $16$. We find that the worst-case throughput is achieved under either chessboard or $U+P$ $0.5$ demand matrices. Interestingly, Figure~\ref{fig:worst-case} shows that demand-aware periodic networks improve the throughput by $30\%$ (absolute) consistently, compared to demand-oblivious networks. Further, the worst-case throughput being consistently close to $0.8$ in our matrices (although limited), gives further hope for closing the gap between our lower bound $\frac{2}{3}$ and upper bound $\frac{4}{5}$.

\begin{figure}[t]
        \centering
        \includegraphics[trim={0 0.5cm 0 1.5cm},clip,width=0.9\linewidth]{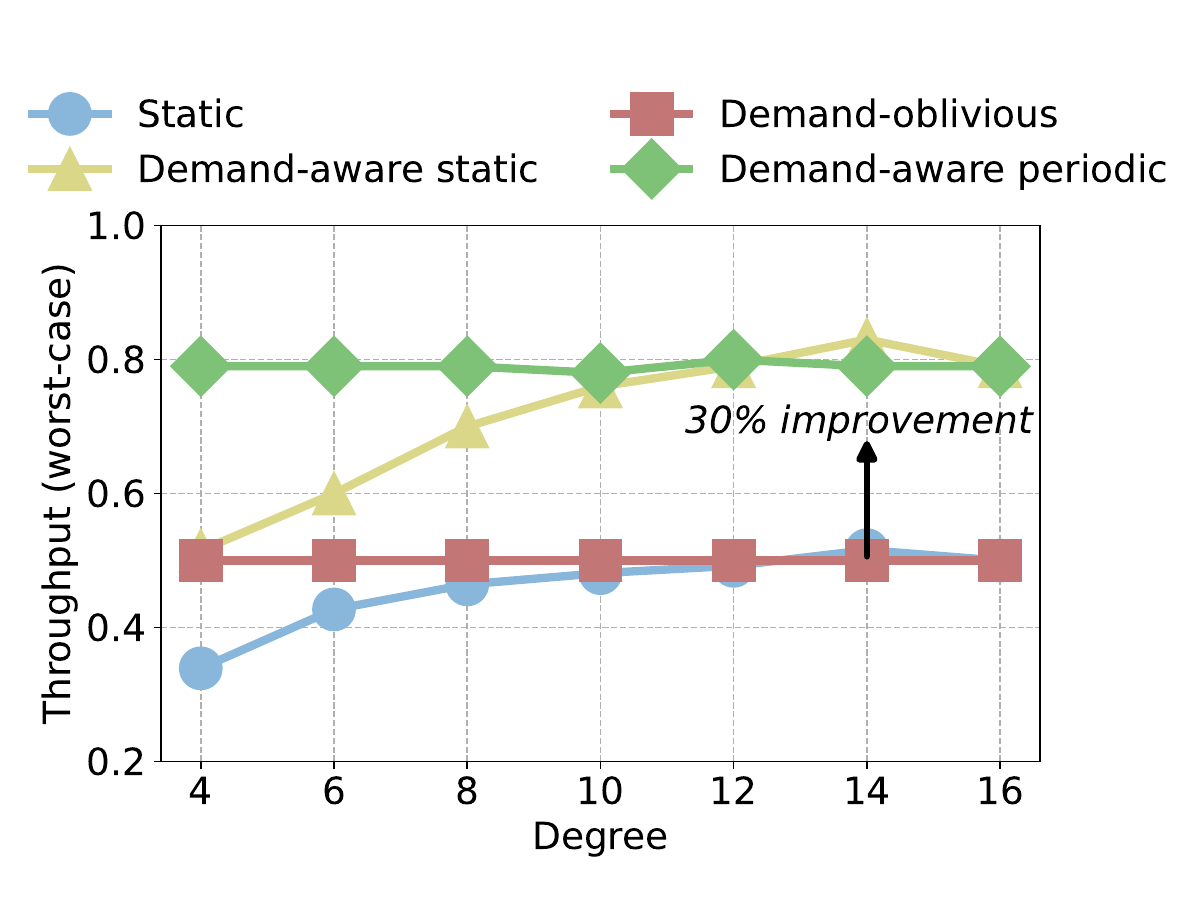}
        \vspace{-3mm}
        \caption{The worst-case throughput of demand-aware periodic is independent of degree and $30\%$ greater than that of demand-oblivious. The throughput of demand-aware static is dependent on degree but close to demand-oblivious even at low degree.}
        \vspace{-4mm}
        \label{fig:worst-case}
\end{figure}

\myitem{Worst-case throughput of demand-aware static depends on degree:} Our analysis in \S\ref{sec:demand-aware-static} establishes the throughput bounds for demand-aware static in the special case of degree $n$. Unsurprisingly, our empirical results in Figure~\ref{fig:worst-case} show that the throughput of demand-aware static is dependent on the degree. The worst-case throughput at low degree ($=4$) is similar to that of demand-oblivious networks. Interestingly, at degree $=4$ (Figure~\ref{fig:degree-4}), the worst-case throughput is achieved under chessboard demand matrix but demand-oblivious networks perform much better for chessboard matrix. Yet, the worst-case throughput for demand-oblivious networks is achieved under permutation demand matrix but demand-aware static networks perform optimally under permutation matrix. As the degree increases, the worst-case throughput converges to that of demand-aware periodic.

\noindent \textbf{Demand-aware static suits ML workloads even with low degree:} Although from Figure~\ref{fig:worst-case}, the worst-case throughput of demand-aware static is low for degree $=4$, the throughput for specific demand matrices is much higher. From Figure~\ref{fig:degree-4}, demand-aware static of degree $4$ improves throughput by $16.4\%$ for data parallelism, by $9.1\%$ for hybrid parallelism, by $9.5\%$ for DLRM $+3$ permuted, and by $23.4\%$ for DLRM $+7$ permuted, compared to demand-oblivious. 

Overall, our empirical results on the throughput of different networks under various demand matrices align with the theoretical bounds presented in \S\ref{sec:throughput}.

\section{Related Work}
Datacenter topologies have been widely studied in the literature both in the context of traditional packet-switched networks~\cite{10.1145/2999572.2999580,180604,10.1145/1402958.1402967,10.1145/1592568.1592576,227667,10.1145/2785956.2787508,7013016,f10,10.1145/1592568.1592577} and emerging reconfigurable optically circuit-switched networks~\cite{10.1145/3098822.3098838,10.1145/3387514.3406221,10.1145/2934872.2934911,10.1145/1851182.1851223,10.1145/3579449,10.1145/2377677.2377761,kandula2009flyways,opera,10.1145/2619239.2626328,201560,10.1145/2619239.2626332,6490069,10.1145/1851182.1851222,7066977,10.1145/2896377.2901479,10.1145/1868447.1868455,10.1145/3491050,278374,10.1145/3351452.3351464}. In the design of topologies, various metrics of interest have been considered. For instance, uniformly high bandwidth availability~\cite{10.1145/1402958.1402967,10.1145/1592568.1592577}, expansion~\cite{10.1145/2999572.2999580,180604}, fault-tolerance~\cite{f10}, and even the life cycle management of a datacenter~\cite{227667}. In the context of reconfigurable networks, typically, the goal has been either to minimize the reconfiguration overhead~\cite{10.1145/3098822.3098838,10.1145/3387514.3406221} or to minimize the bandwidth tax~\cite{10.1145/2934872.2934911,10.1145/1851182.1851223,10.1145/3579449,285119}.

Bisection bandwidth has been extensively used as metric for topologies in the past in order to reason about the capacity and potential bottlenecks of a topology. Recent works argue for a new measure \ie ``throughput'', to understand the maximum load supported by a topology~\cite{7877143,179775,10.1145/3452296.3472913,10.1145/3491050,10.1145/3579312}. In fact, the max-flow that relates to the throughput of a topology, can be $\mathcal{O}(\log(n))$ factor lower than the sparsest cut~\cite{10.1145/331524.331526,10.1109/SFCS.1988.21958,7877143}. Namyar et al.~study the throughput upper bound for static datacenter topologies and show a separation between Clos (i.e., fat-trees) and expander-based networks in terms of throughput~\cite{10.1145/3452296.3472913}. In the context of reconfigurable networks, only recently have the throughput bounds of demand-oblivious networks been established~\cite{10.1145/3579312,10.1145/3491050,10.1145/3519935.3520020}. In this paper, we focus on the throughput landscape of reconfigurable networks as a whole, showing a separation result between demand-aware and demand-oblivious networks.

While throughput of a datacenter topology is interesting from a theory standpoint, a vast majority of the literature focuses on practically achieving the ideal throughput of a topology. For instance, congestion control~\cite{10.1145/1851182.1851192,10.1145/2785956.2787510,10.1145/3341302.3342085,278346,10.1145/3387514.3406591,276958,10.1145/3387514.3405899,10.1145/2018436.2018443}, buffer management~\cite{abm,fab,trafficaware,295539,10229046,295535}, scheduling~\cite{10.1145/2486001.2486031,259355,cassini}, load-balancing~\cite{10.1145/2619239.2626316,10.1145/2890955.2890968,10.1145/3098822.3098839}. In fact, the underlying protocols can turn out to be the key enablers (or limiters) of system performance in the datacenter~\cite{10.1145/3387514.3406591}. Only recently, congestion control tailored for reconfigurable networks has been considered~\cite{10.1145/3603269.3610862,10.1145/3544216.3544254,246336}. We leave it for future work to study the protocols required by reconfigurable networks to reach their ideal throughput.

\section{Future Research Directions}
The primary objective of this paper has been to explore the \emph{throughput} landscape of reconfigurable datacenter networks. We believe this paper opens several interesting avenues for future work, encompassing both systems and theoretical aspects. In this section, we briefly outline some of these prospective research directions \first on the practical realization of the theoretical throughput of such networks and \second on enhancing the throughput bounds established in this paper.

\vspace{-3mm}
\subsection{Systems}
Our analysis in this paper focuses on the throughput that is \emph{ideally} achievable by reconfigurable networks. To achieve this ideal throughput, various protocols need to function together optimally, especially routing and congestion control.

\myitem{Routing:} 
Traditional datacenter networks predominantly utilize equal-cost multipath (ECMP) routing, often at per-flow granularity. However, in the realm of reconfigurable networks, ECMP does not suffice to maximize throughput due to the presence of multiple paths between source-destination pairs that are not necessarily equal in cost \ie single-hop paths (direct-connect circuits) and multi-hop paths (via intermediate nodes). For example, in demand-oblivious reconfigurable networks, a single-hop path is available only once in a given period, necessitating the use of $2$-hop paths as well to improve throughput~\cite{10.1145/3098822.3098838, 10.1145/3387514.3406221, 10.1145/3579312}. Previous studies, especially those focusing on demand-oblivious networks, advocate for Valiant routing~\cite{valiant1982scheme}, a demand-oblivious routing scheme. However, this approach can reduce throughput by a factor of $2$ in the worst-case scenario.
While demand-aware routing (\ie adaptive routing) algorithms have been proposed for periodic reconfigurable networks~\cite{10.1145/3098822.3098838,opera}, they are limited to $2$-hop paths.
We believe two areas of research on routing can significantly enhance the practicality and the throughput benefits of reconfigurable networks. Firstly, generalizing existing demand-aware routing algorithms in reconfigurable networks to incorporate $k$-shortest paths. Secondly, understanding the convergence time of demand-aware routing algorithms in response to changes in the switching schedule.

\myitem{Congestion control:} Reconfigurable networks present with unique set of challenges for congestion control. The available bandwidth can change drastically after a reconfiguration~\cite{246336,10.1145/3544216.3544254}. A positive queue length does not necessarily imply $100\%$ link utilization, in fact, a positive queue length could also imply zero utilization (due to waiting times for next-hop circuit). Interestingly, the throughput of certain reconfigurable networks also depends on the available buffers in the network~\cite{10.1145/3579312}.
While recent works propose congestion control algorithms for reconfigurable networks, they are still limited to either $2$-hop paths~\cite{10.1145/3387514.3406221,10.1145/3098822.3098838} or periodic reconfigurations~\cite{10.1145/3544216.3544254,246336,10.1145/3603269.3610862}. Future research on congestion control algorithms, suitable for a wide spectrum of reconfigurable networks, would not only enhance the practically achievable throughput of these networks, but also facilitates realistic packet-level simulations (e.g., in NS3~\cite{ns3}).

% Queue buildup has been historically viewed as an indicator for congestion especially in the context of datacenters, where near-zero queue lengths are desirable. End-host-based congestion control protocols perceive network congestion throughput various signals such as ECN~\cite{10.1145/1851182.1851192}, RTT~\cite{10.1145/2785956.2787510,10.1145/3387514.3406591}, INT~\cite{278346}, all of which react to increased queue lengths. Interestingly, periodic reconfigurable networks require certain amount of buffering at each node in order to sustain full-throughput~\cite{10.1145/3579312}. While prior works proposed congestion control algorithms for such networks, they are still limited to either $2$-hop paths~\cite{10.1145/3387514.3406221} or periodic reconfigurations~\cite{10.1145/3544216.3544254,246336}. We believe reconfigurable networks do not fundamentally align with the design principles of TCP and the notion of ``pipe'' and ``stream''.

% However, circuit switches require certain amount of time to reconfigure and the links remains inactive during this period, implying the need for buffering packets at the downstream switch or server.

\myitem{Throughput per $\$$ analysis:} This paper does not engage in comparisons with topologies constructed using electrical packet switches (e.g., Clos-based~\cite{10.1145/1402958.1402967}). The common approach to estimate cost, especially in systems evaluations is to scale the cost linearly by the number of ports and cables used in the network~\cite{opera,10.1145/3098822.3098838}. Yet, from a throughput perspective, the comparison of throughput achieved per unit cost would change significantly between almost linear, superlinear, and sublinear cost functions.
A fair comparison between packet-switched and circuit-switched topologies in a formal setting necessitates well-defined cost functions for the switches. Future research efforts aimed at developing cost functions would open up interesting avenues, such as formally studying throughput per cost as a metric to compare different datacenter topologies.

\subsection{Theory}
The theoretical results in this paper provide insights into the landscape of throughput bounds of reconfigurable networks. Our analytical framework features interesting connections to classic problems in the literature, opening opportunities for future research directions to tighten our bounds and to generalize the results.

\myitem{Connections to matrix rounding problem:}
As observed in this paper, demand-aware topology design is intuitively related to the integer-residual matrix decomposition of the demand matrix. We use floor function for matrix decomposition in our proofs. An alternative method is \emph{matrix rounding}, which involves adjusting each entry of the matrix by either applying the floor or ceiling function in such a way that the sums of the rows and columns remain unchanged~\cite{bacharach1966matrix}. Interestingly, there always exists a feasible solution for matrix rounding based on a formlution of maximum interger flow in a network specified by rows and colums of the matrix~\cite{bacharach1966matrix}. In other words, given a saturated demand matrix (doubly stochastic), the solution to matrix rounding gives the integer part of our decomposition without changing the row and column sums. This implies that a rounded matrix is always regular \ie adding demand-aware links based on the rounded matrix results in a regular topology --- a desirable property for common graph theoretic techniques, especially for throughput.
We believe that drawing more insights from matrix rounding problem in the future can potentially tighten our bounds and generalize our results to the set of all demand matrices within the hose model.

\myitem{Understanding the latency of demand-aware networks:} Recent works formally show inherent tradeoffs in demand-oblivious reconfigurable networks~\cite{10.1145/3519935.3520020,10.1145/3579312}. Our focus in this paper is on the throughput landscape of reconfigurable networks. It is an interesting future research direction to formally study the landscape under a joint-objective between throughput, latency and buffer requirements. 
% \blindtext
% \blindtext
% \blindtext
% \blindtext
% \blindtext
% \blindtext
% \blindtext

\section{Conclusion}
We presented the throughput landscape of reconfigurable networks, formally establishing a clear distinction between demand-oblivious and demand-aware reconfigurable networks. We presented both upper and lower bounds for the throughput of demand-aware networks.
Our analytical framework allowed us to unveil innovative reconfigurable network designs that combine the simplicity of circuit switching characteristic of demand-oblivious networks with the throughput advantages inherent to demand-aware networks.
In the future, we plan to formally study the two-dimensional landscape encompassing both  throughput and latency in reconfigurable networks.
\label{bodyLastPage}

\balance
\clearpage

{ \balance
\bibliographystyle{unsrt}
\bibliography{references}

\begin{thebibliography}{10}

\bibitem{10.1145/2785956.2787508}
Arjun Singh, Joon Ong, Amit Agarwal, Glen Anderson, Ashby Armistead, Roy
  Bannon, Seb Boving, Gaurav Desai, Bob Felderman, Paulie Germano, Anand
  Kanagala, Jeff Provost, Jason Simmons, Eiichi Tanda, Jim Wanderer, Urs
  H\"{o}lzle, Stephen Stuart, and Amin Vahdat.
\newblock Jupiter rising: A decade of clos topologies and centralized control
  in google's datacenter network.
\newblock In {\em Proceedings of the 2015 ACM Conference on Special Interest
  Group on Data Communication}, SIGCOMM '15, page 183–197, New York, NY, USA,
  2015. Association for Computing Machinery.

\bibitem{10.1145/3544216.3544265}
Leon Poutievski, Omid Mashayekhi, Joon Ong, Arjun Singh, Mukarram Tariq, Rui
  Wang, Jianan Zhang, Virginia Beauregard, Patrick Conner, Steve Gribble, Rishi
  Kapoor, Stephen Kratzer, Nanfang Li, Hong Liu, Karthik Nagaraj, Jason
  Ornstein, Samir Sawhney, Ryohei Urata, Lorenzo Vicisano, Kevin Yasumura,
  Shidong Zhang, Junlan Zhou, and Amin Vahdat.
\newblock Jupiter evolving: transforming google's datacenter network via
  optical circuit switches and software-defined networking.
\newblock In {\em Proceedings of the ACM SIGCOMM 2022 Conference}, SIGCOMM '22,
  page 66–85, New York, NY, USA, 2022. Association for Computing Machinery.

\bibitem{10.1145/3387514.3406221}
Hitesh Ballani, Paolo Costa, Raphael Behrendt, Daniel Cletheroe, Istvan Haller,
  Krzysztof Jozwik, Fotini Karinou, Sophie Lange, Kai Shi, Benn Thomsen, and
  Hugh Williams.
\newblock Sirius: A flat datacenter network with nanosecond optical switching.
\newblock In {\em Proceedings of the Annual Conference of the ACM Special
  Interest Group on Data Communication on the Applications, Technologies,
  Architectures, and Protocols for Computer Communication}, SIGCOMM '20, page
  782–797, New York, NY, USA, 2020. Association for Computing Machinery.

\bibitem{10.1145/3098822.3098838}
William~M. Mellette, Rob McGuinness, Arjun Roy, Alex Forencich, George Papen,
  Alex~C. Snoeren, and George Porter.
\newblock Rotornet: A scalable, low-complexity, optical datacenter network.
\newblock In {\em Proceedings of the Conference of the ACM Special Interest
  Group on Data Communication}, SIGCOMM '17, page 267–280, New York, NY, USA,
  2017. Association for Computing Machinery.

\bibitem{10.1145/2934872.2934911}
Monia Ghobadi, Ratul Mahajan, Amar Phanishayee, Nikhil Devanur, Janardhan
  Kulkarni, Gireeja Ranade, Pierre-Alexandre Blanche, Houman Rastegarfar,
  Madeleine Glick, and Daniel Kilper.
\newblock Projector: Agile reconfigurable data center interconnect.
\newblock In {\em Proceedings of the 2016 ACM SIGCOMM Conference}, SIGCOMM '16,
  page 216–229, New York, NY, USA, 2016. Association for Computing Machinery.

\bibitem{10.1145/1851182.1851223}
Nathan Farrington, George Porter, Sivasankar Radhakrishnan, Hamid~Hajabdolali
  Bazzaz, Vikram Subramanya, Yeshaiahu Fainman, George Papen, and Amin Vahdat.
\newblock Helios: a hybrid electrical/optical switch architecture for modular
  data centers.
\newblock In {\em Proceedings of the ACM SIGCOMM 2010 Conference}, SIGCOMM '10,
  page 339–350, New York, NY, USA, 2010. Association for Computing Machinery.

\bibitem{10.1145/3579312}
Vamsi Addanki, Chen Avin, and Stefan Schmid.
\newblock Mars: Near-optimal throughput with shallow buffers in reconfigurable
  datacenter networks.
\newblock {\em Proc. ACM Meas. Anal. Comput. Syst.}, 7(1), mar 2023.

\bibitem{10.1145/3579449}
Johannes Zerwas, Csaba Gy\"{o}rgyi, Andreas Blenk, Stefan Schmid, and Chen
  Avin.
\newblock Duo: A high-throughput reconfigurable datacenter network using local
  routing and control.
\newblock {\em Proc. ACM Meas. Anal. Comput. Syst.}, 7(1), mar 2023.

\bibitem{opera}
William~M. Mellette, Rajdeep Das, Yibo Guo, Rob McGuinness, Alex~C. Snoeren,
  and George Porter.
\newblock Expanding across time to deliver bandwidth efficiency and low
  latency.
\newblock In {\em 17th USENIX Symposium on Networked Systems Design and
  Implementation (NSDI 20)}, pages 1--18, Santa Clara, CA, February 2020.
  USENIX Association.

\bibitem{285119}
Weiyang Wang, Moein Khazraee, Zhizhen Zhong, Manya Ghobadi, Zhihao Jia,
  Dheevatsa Mudigere, Ying Zhang, and Anthony Kewitsch.
\newblock {TopoOpt}: Co-optimizing network topology and parallelization
  strategy for distributed training jobs.
\newblock In {\em 20th USENIX Symposium on Networked Systems Design and
  Implementation (NSDI 23)}, pages 739--767, Boston, MA, April 2023. USENIX
  Association.

\bibitem{10.1145/3519935.3520020}
Daniel Amir, Tegan Wilson, Vishal Shrivastav, Hakim Weatherspoon, Robert
  Kleinberg, and Rachit Agarwal.
\newblock Optimal oblivious reconfigurable networks.
\newblock In {\em Proceedings of the 54th Annual ACM SIGACT Symposium on Theory
  of Computing}, STOC 2022, page 1339–1352, New York, NY, USA, 2022.
  Association for Computing Machinery.

\bibitem{10.1145/3491050}
Chen Griner, Johannes Zerwas, Andreas Blenk, Manya Ghobadi, Stefan Schmid, and
  Chen Avin.
\newblock Cerberus: The power of choices in datacenter topology design - a
  throughput perspective.
\newblock {\em Proc. ACM Meas. Anal. Comput. Syst.}, 5(3), dec 2021.

\bibitem{10.1145/3452296.3472913}
Pooria Namyar, Sucha Supittayapornpong, Mingyang Zhang, Minlan Yu, and Ramesh
  Govindan.
\newblock A throughput-centric view of the performance of datacenter
  topologies.
\newblock In {\em Proceedings of the 2021 ACM SIGCOMM 2021 Conference}, SIGCOMM
  '21, page 349–369, New York, NY, USA, 2021. Association for Computing
  Machinery.

\bibitem{10.1145/316188.316209}
N.~G. Duffield, Pawan Goyal, Albert Greenberg, Partho Mishra, K.~K.
  Ramakrishnan, and Jacobus~E. van~der Merive.
\newblock A flexible model for resource management in virtual private networks.
\newblock In {\em Proceedings of the Conference on Applications, Technologies,
  Architectures, and Protocols for Computer Communication}, SIGCOMM '99, page
  95–108, New York, NY, USA, 1999. Association for Computing Machinery.

\bibitem{7877143}
Sangeetha~Abdu Jyothi, Ankit Singla, P.~Brighten Godfrey, and Alexandra Kolla.
\newblock Measuring and understanding throughput of network topologies.
\newblock In {\em SC '16: Proceedings of the International Conference for High
  Performance Computing, Networking, Storage and Analysis}, pages 761--772,
  2016.

\bibitem{doi:10.1177/1094342005051521}
Rajeev Thakur, Rolf Rabenseifner, and William Gropp.
\newblock Optimization of collective communication operations in mpich.
\newblock {\em The International Journal of High Performance Computing
  Applications}, 19(1):49--66, 2005.

\bibitem{NIPS2014_7d6044e9}
Seunghak Lee, Jin~Kyu Kim, Xun Zheng, Qirong Ho, Garth~A Gibson, and Eric~P
  Xing.
\newblock On model parallelization and scheduling strategies for distributed
  machine learning.
\newblock In Z.~Ghahramani, M.~Welling, C.~Cortes, N.~Lawrence, and K.Q.
  Weinberger, editors, {\em Advances in Neural Information Processing Systems},
  volume~27. Curran Associates, Inc., 2014.

\bibitem{10.1145/2716281.2836126}
He~Liu, Matthew~K. Mukerjee, Conglong Li, Nicolas Feltman, George Papen, Stefan
  Savage, Srinivasan Seshan, Geoffrey~M. Voelker, David~G. Andersen, Michael
  Kaminsky, George Porter, and Alex~C. Snoeren.
\newblock Scheduling techniques for hybrid circuit/packet networks.
\newblock In {\em Proceedings of the 11th ACM Conference on Emerging Networking
  Experiments and Technologies}, CoNEXT '15, New York, NY, USA, 2015.
  Association for Computing Machinery.

\bibitem{10.1145/2486001.2486007}
George Porter, Richard Strong, Nathan Farrington, Alex Forencich, Pang
  Chen-Sun, Tajana Rosing, Yeshaiahu Fainman, George Papen, and Amin Vahdat.
\newblock Integrating microsecond circuit switching into the data center.
\newblock In {\em Proceedings of the ACM SIGCOMM 2013 Conference on SIGCOMM},
  SIGCOMM '13, page 447–458, New York, NY, USA, 2013. Association for
  Computing Machinery.

\bibitem{birkhoff1946three}
Garrett Birkhoff.
\newblock Three observations on linear algebra.
\newblock {\em Univ. Nac. Tacuman, Rev. Ser. A}, 5:147--151, 1946.

\bibitem{gurobi}
{Gurobi Optimization, LLC}.
\newblock {Gurobi Optimizer Reference Manual}, 2023.

\bibitem{10.1145/2999572.2999580}
Asaf Valadarsky, Gal Shahaf, Michael Dinitz, and Michael Schapira.
\newblock Xpander: Towards optimal-performance datacenters.
\newblock In {\em Proceedings of the 12th International on Conference on
  Emerging Networking EXperiments and Technologies}, CoNEXT '16, page
  205–219, New York, NY, USA, 2016. Association for Computing Machinery.

\bibitem{180604}
Ankit Singla, Chi-Yao Hong, Lucian Popa, and P.~Brighten Godfrey.
\newblock Jellyfish: Networking data centers randomly.
\newblock In {\em 9th USENIX Symposium on Networked Systems Design and
  Implementation (NSDI 12)}, pages 225--238, San Jose, CA, April 2012. USENIX
  Association.

\bibitem{10.1145/1402958.1402967}
Mohammad Al-Fares, Alexander Loukissas, and Amin Vahdat.
\newblock A scalable, commodity data center network architecture.
\newblock In {\em Proceedings of the ACM SIGCOMM 2008 Conference on Data
  Communication}, SIGCOMM '08, page 63–74, New York, NY, USA, 2008.
  Association for Computing Machinery.

\bibitem{10.1145/1592568.1592576}
Albert Greenberg, James~R. Hamilton, Navendu Jain, Srikanth Kandula, Changhoon
  Kim, Parantap Lahiri, David~A. Maltz, Parveen Patel, and Sudipta Sengupta.
\newblock Vl2: a scalable and flexible data center network.
\newblock In {\em Proceedings of the ACM SIGCOMM 2009 Conference on Data
  Communication}, SIGCOMM '09, page 51–62, New York, NY, USA, 2009.
  Association for Computing Machinery.

\bibitem{227667}
Mingyang Zhang, Radhika~Niranjan Mysore, Sucha Supittayapornpong, and Ramesh
  Govindan.
\newblock Understanding lifecycle management complexity of datacenter
  topologies.
\newblock In {\em 16th USENIX Symposium on Networked Systems Design and
  Implementation (NSDI 19)}, pages 235--254, Boston, MA, February 2019. USENIX
  Association.

\bibitem{7013016}
Maciej Besta and Torsten Hoefler.
\newblock Slim fly: A cost effective low-diameter network topology.
\newblock In {\em SC '14: Proceedings of the International Conference for High
  Performance Computing, Networking, Storage and Analysis}, pages 348--359,
  2014.

\bibitem{f10}
Vincent Liu, Daniel Halperin, Arvind Krishnamurthy, and Thomas Anderson.
\newblock F10: A {{Fault-Tolerant}} engineered network.
\newblock In {\em 10th USENIX Symposium on Networked Systems Design and
  Implementation (NSDI 13)}, pages 399--412, Lombard, IL, April 2013. USENIX
  Association.

\bibitem{10.1145/1592568.1592577}
Chuanxiong Guo, Guohan Lu, Dan Li, Haitao Wu, Xuan Zhang, Yunfeng Shi, Chen
  Tian, Yongguang Zhang, and Songwu Lu.
\newblock Bcube: a high performance, server-centric network architecture for
  modular data centers.
\newblock In {\em Proceedings of the ACM SIGCOMM 2009 Conference on Data
  Communication}, SIGCOMM '09, page 63–74, New York, NY, USA, 2009.
  Association for Computing Machinery.

\bibitem{10.1145/2377677.2377761}
Xia Zhou, Zengbin Zhang, Yibo Zhu, Yubo Li, Saipriya Kumar, Amin Vahdat, Ben~Y.
  Zhao, and Haitao Zheng.
\newblock Mirror mirror on the ceiling: flexible wireless links for data
  centers.
\newblock {\em SIGCOMM Comput. Commun. Rev.}, 42(4):443–454, aug 2012.

\bibitem{kandula2009flyways}
Srikanth Kandula, Jitendra Padhye, and Paramvir Bahl.
\newblock Flyways to de-congest data center networks.
\newblock In {\em HotNets}. {ACM} {SIGCOMM}, 2009.

\bibitem{10.1145/2619239.2626328}
Navid Hamedazimi, Zafar Qazi, Himanshu Gupta, Vyas Sekar, Samir~R. Das, Jon~P.
  Longtin, Himanshu Shah, and Ashish Tanwer.
\newblock Firefly: a reconfigurable wireless data center fabric using
  free-space optics.
\newblock In {\em Proceedings of the 2014 ACM Conference on SIGCOMM}, SIGCOMM
  '14, page 319–330, New York, NY, USA, 2014. Association for Computing
  Machinery.

\bibitem{201560}
Li~Chen, Kai Chen, Zhonghua Zhu, Minlan Yu, George Porter, Chunming Qiao, and
  Shan Zhong.
\newblock Enabling {Wide-Spread} communications on optical fabric with
  {MegaSwitch}.
\newblock In {\em 14th USENIX Symposium on Networked Systems Design and
  Implementation (NSDI 17)}, pages 577--593, Boston, MA, March 2017. USENIX
  Association.

\bibitem{10.1145/2619239.2626332}
Yunpeng~James Liu, Peter~Xiang Gao, Bernard Wong, and Srinivasan Keshav.
\newblock Quartz: a new design element for low-latency dcns.
\newblock In {\em Proceedings of the 2014 ACM Conference on SIGCOMM}, SIGCOMM
  '14, page 283–294, New York, NY, USA, 2014. Association for Computing
  Machinery.

\bibitem{6490069}
Kai Chen, Ankit Singla, Atul Singh, Kishore Ramachandran, Lei Xu, Yueping
  Zhang, Xitao Wen, and Yan Chen.
\newblock Osa: An optical switching architecture for data center networks with
  unprecedented flexibility.
\newblock {\em IEEE/ACM Transactions on Networking}, 22(2):498--511, 2014.

\bibitem{10.1145/1851182.1851222}
Guohui Wang, David~G. Andersen, Michael Kaminsky, Konstantina Papagiannaki,
  T.S.~Eugene Ng, Michael Kozuch, and Michael Ryan.
\newblock c-through: part-time optics in data centers.
\newblock In {\em Proceedings of the ACM SIGCOMM 2010 Conference}, SIGCOMM '10,
  page 327–338, New York, NY, USA, 2010. Association for Computing Machinery.

\bibitem{7066977}
Stefan Schmid, Chen Avin, Christian Scheideler, Michael Borokhovich, Bernhard
  Haeupler, and Zvi Lotker.
\newblock Splaynet: Towards locally self-adjusting networks.
\newblock {\em IEEE/ACM Transactions on Networking}, 24(3):1421--1433, 2016.

\bibitem{10.1145/2896377.2901479}
Shaileshh Bojja~Venkatakrishnan, Mohammad Alizadeh, and Pramod Viswanath.
\newblock Costly circuits, submodular schedules and approximate
  carath\'{e}odory theorems.
\newblock In {\em Proceedings of the 2016 ACM SIGMETRICS International
  Conference on Measurement and Modeling of Computer Science}, SIGMETRICS '16,
  page 75–88, New York, NY, USA, 2016. Association for Computing Machinery.

\bibitem{10.1145/1868447.1868455}
Ankit Singla, Atul Singh, Kishore Ramachandran, Lei Xu, and Yueping Zhang.
\newblock Proteus: a topology malleable data center network.
\newblock In {\em Proceedings of the 9th ACM SIGCOMM Workshop on Hot Topics in
  Networks}, Hotnets-IX, New York, NY, USA, 2010. Association for Computing
  Machinery.

\bibitem{278374}
Weitao Wang, Dingming Wu, Sushovan Das, Afsaneh Rahbar, Ang Chen, and
  T.~S.~Eugene Ng.
\newblock {RDC}: {Energy-Efficient} data center network congestion relief with
  topological reconfigurability at the edge.
\newblock In {\em 19th USENIX Symposium on Networked Systems Design and
  Implementation (NSDI 22)}, pages 1267--1288, Renton, WA, April 2022. USENIX
  Association.

\bibitem{10.1145/3351452.3351464}
Klaus-Tycho Foerster and Stefan Schmid.
\newblock Survey of reconfigurable data center networks: Enablers, algorithms,
  complexity.
\newblock {\em SIGACT News}, 50(2):62–79, jul 2019.

\bibitem{179775}
Ankit Singla, P.~Brighten Godfrey, and Alexandra Kolla.
\newblock High throughput data center topology design.
\newblock In {\em 11th USENIX Symposium on Networked Systems Design and
  Implementation (NSDI 14)}, pages 29--41, Seattle, WA, April 2014. USENIX
  Association.

\bibitem{10.1145/331524.331526}
Tom Leighton and Satish Rao.
\newblock Multicommodity max-flow min-cut theorems and their use in designing
  approximation algorithms.
\newblock {\em J. ACM}, 46(6):787–832, nov 1999.

\bibitem{10.1109/SFCS.1988.21958}
T.~Leighton and S.~Rao.
\newblock An approximate max-flow min-cut theorem for uniform multicommodity
  flow problems with applications to approximation algorithms.
\newblock In {\em Proceedings of the 29th Annual Symposium on Foundations of
  Computer Science}, SFCS '88, page 422–431, USA, 1988. IEEE Computer
  Society.

\bibitem{10.1145/1851182.1851192}
Mohammad Alizadeh, Albert Greenberg, David~A. Maltz, Jitendra Padhye, Parveen
  Patel, Balaji Prabhakar, Sudipta Sengupta, and Murari Sridharan.
\newblock Data center tcp (dctcp).
\newblock In {\em Proceedings of the ACM SIGCOMM 2010 Conference}, SIGCOMM '10,
  page 63–74, New York, NY, USA, 2010. Association for Computing Machinery.

\bibitem{10.1145/2785956.2787510}
Radhika Mittal, Vinh~The Lam, Nandita Dukkipati, Emily Blem, Hassan Wassel,
  Monia Ghobadi, Amin Vahdat, Yaogong Wang, David Wetherall, and David Zats.
\newblock Timely: Rtt-based congestion control for the datacenter.
\newblock In {\em Proceedings of the 2015 ACM Conference on Special Interest
  Group on Data Communication}, SIGCOMM '15, page 537–550, New York, NY, USA,
  2015. Association for Computing Machinery.

\bibitem{10.1145/3341302.3342085}
Yuliang Li, Rui Miao, Hongqiang~Harry Liu, Yan Zhuang, Fei Feng, Lingbo Tang,
  Zheng Cao, Ming Zhang, Frank Kelly, Mohammad Alizadeh, and Minlan Yu.
\newblock Hpcc: High precision congestion control.
\newblock In {\em Proceedings of the ACM Special Interest Group on Data
  Communication}, SIGCOMM '19, page 44–58, New York, NY, USA, 2019.
  Association for Computing Machinery.

\bibitem{278346}
Vamsi Addanki, Oliver Michel, and Stefan Schmid.
\newblock {PowerTCP}: Pushing the performance limits of datacenter networks.
\newblock In {\em 19th USENIX Symposium on Networked Systems Design and
  Implementation (NSDI 22)}, pages 51--70, Renton, WA, April 2022. USENIX
  Association.

\bibitem{10.1145/3387514.3406591}
Gautam Kumar, Nandita Dukkipati, Keon Jang, Hassan M.~G. Wassel, Xian Wu,
  Behnam Montazeri, Yaogong Wang, Kevin Springborn, Christopher Alfeld, Michael
  Ryan, David Wetherall, and Amin Vahdat.
\newblock Swift: Delay is simple and effective for congestion control in the
  datacenter.
\newblock In {\em Proceedings of the Annual Conference of the ACM Special
  Interest Group on Data Communication on the Applications, Technologies,
  Architectures, and Protocols for Computer Communication}, SIGCOMM '20, page
  514–528, New York, NY, USA, 2020. Association for Computing Machinery.

\bibitem{276958}
Prateesh Goyal, Preey Shah, Kevin Zhao, Georgios Nikolaidis, Mohammad Alizadeh,
  and Thomas~E. Anderson.
\newblock Backpressure flow control.
\newblock In {\em 19th USENIX Symposium on Networked Systems Design and
  Implementation (NSDI 22)}, pages 779--805, Renton, WA, April 2022. USENIX
  Association.

\bibitem{10.1145/3387514.3405899}
Ahmed Saeed, Varun Gupta, Prateesh Goyal, Milad Sharif, Rong Pan, Mostafa
  Ammar, Ellen Zegura, Keon Jang, Mohammad Alizadeh, Abdul Kabbani, and Amin
  Vahdat.
\newblock Annulus: A dual congestion control loop for datacenter and wan
  traffic aggregates.
\newblock SIGCOMM '20, page 735–749, New York, NY, USA, 2020. Association for
  Computing Machinery.

\bibitem{10.1145/2018436.2018443}
Christo Wilson, Hitesh Ballani, Thomas Karagiannis, and Ant Rowtron.
\newblock Better never than late: meeting deadlines in datacenter networks.
\newblock In {\em Proceedings of the ACM SIGCOMM 2011 Conference}, SIGCOMM '11,
  page 50–61, New York, NY, USA, 2011. Association for Computing Machinery.

\bibitem{abm}
Vamsi Addanki, Maria Apostolaki, Manya Ghobadi, Stefan Schmid, and Laurent
  Vanbever.
\newblock Abm: Active buffer management in datacenters.
\newblock In {\em Proceedings of the ACM SIGCOMM 2022 Conference}, SIGCOMM '22,
  page 36–52, New York, NY, USA, 2022. Association for Computing Machinery.

\bibitem{fab}
Maria Apostolaki, Laurent Vanbever, and Manya Ghobadi.
\newblock Fab: Toward flow-aware buffer sharing on programmable switches.
\newblock In {\em Proceedings of the 2019 Workshop on Buffer Sizing}, BS '19,
  New York, NY, USA, 2020. Association for Computing Machinery.

\bibitem{trafficaware}
Sijiang Huang, Mowei Wang, and Yong Cui.
\newblock Traffic-aware buffer management in shared memory switches.
\newblock {\em IEEE/ACM Transactions on Networking}, 30(6):2559--2573, 2022.

\bibitem{295539}
Vamsi Addanki, Wei Bai, Stefan Schmid, and Maria Apostolaki.
\newblock Reverie: Low pass {Filter-Based} switch buffer sharing for
  datacenters with {RDMA} and {TCP} traffic.
\newblock In {\em 21st USENIX Symposium on Networked Systems Design and
  Implementation (NSDI 24)}, pages 651--668, Santa Clara, CA, April 2024.
  USENIX Association.

\bibitem{10229046}
Hamidreza Almasi, Rohan Vardekar, and Balajee Vamanan.
\newblock Protean: Adaptive management of shared-memory in datacenter switches.
\newblock In {\em IEEE INFOCOM 2023 - IEEE Conference on Computer
  Communications}, pages 1--10, 2023.

\bibitem{295535}
Vamsi Addanki, Maciej Pacut, and Stefan Schmid.
\newblock Credence: Augmenting datacenter switch buffer sharing with {ML}
  predictions.
\newblock In {\em 21st USENIX Symposium on Networked Systems Design and
  Implementation (NSDI 24)}, pages 613--634, Santa Clara, CA, April 2024.
  USENIX Association.

\bibitem{10.1145/2486001.2486031}
Mohammad Alizadeh, Shuang Yang, Milad Sharif, Sachin Katti, Nick McKeown,
  Balaji Prabhakar, and Scott Shenker.
\newblock Pfabric: Minimal near-optimal datacenter transport.
\newblock In {\em Proceedings of the ACM SIGCOMM 2013 Conference on SIGCOMM},
  SIGCOMM '13, page 435–446, New York, NY, USA, 2013. Association for
  Computing Machinery.

\bibitem{259355}
Mohammad Al-Fares, Sivasankar Radhakrishnan, Barath Raghavan, Nelson Huang, and
  Amin Vahdat.
\newblock Hedera: Dynamic flow scheduling for data center networks.
\newblock In {\em 7th USENIX Symposium on Networked Systems Design and
  Implementation (NSDI 10)}, San Jose, CA, April 2010. USENIX Association.

\bibitem{cassini}
Sudarsanan Rajasekaran, Manya Ghobadi, and Aditya Akella.
\newblock {CASSINI}: {Network-Aware} job scheduling in machine learning
  clusters.
\newblock In {\em 21st USENIX Symposium on Networked Systems Design and
  Implementation (NSDI 24)}, pages 1403--1420, Santa Clara, CA, April 2024.
  USENIX Association.

\bibitem{10.1145/2619239.2626316}
Mohammad Alizadeh, Tom Edsall, Sarang Dharmapurikar, Ramanan Vaidyanathan,
  Kevin Chu, Andy Fingerhut, Vinh~The Lam, Francis Matus, Rong Pan, Navindra
  Yadav, and George Varghese.
\newblock Conga: Distributed congestion-aware load balancing for datacenters.
\newblock In {\em Proceedings of the 2014 ACM Conference on SIGCOMM}, SIGCOMM
  '14, page 503–514, New York, NY, USA, 2014. Association for Computing
  Machinery.

\bibitem{10.1145/2890955.2890968}
Naga Katta, Mukesh Hira, Changhoon Kim, Anirudh Sivaraman, and Jennifer
  Rexford.
\newblock Hula: Scalable load balancing using programmable data planes.
\newblock In {\em Proceedings of the Symposium on SDN Research}, SOSR '16, New
  York, NY, USA, 2016. Association for Computing Machinery.

\bibitem{10.1145/3098822.3098839}
Soudeh Ghorbani, Zibin Yang, P.~Brighten Godfrey, Yashar Ganjali, and Amin
  Firoozshahian.
\newblock Drill: Micro load balancing for low-latency data center networks.
\newblock In {\em Proceedings of the Conference of the ACM Special Interest
  Group on Data Communication}, SIGCOMM '17, page 225–238, New York, NY, USA,
  2017. Association for Computing Machinery.

\bibitem{10.1145/3603269.3610862}
Daniel Amir, Tegan Wilson, Vishal Shrivastav, Hakim Weatherspoon, and Robert
  Kleinberg.
\newblock Poster: Scalability and congestion control in oblivious
  reconfigurable networks.
\newblock In {\em Proceedings of the ACM SIGCOMM 2023 Conference}, ACM SIGCOMM
  '23, page 1138–1140, New York, NY, USA, 2023. Association for Computing
  Machinery.

\bibitem{10.1145/3544216.3544254}
Shawn~Shuoshuo Chen, Weiyang Wang, Christopher Canel, Srinivasan Seshan,
  Alex~C. Snoeren, and Peter Steenkiste.
\newblock Time-division tcp for reconfigurable data center networks.
\newblock In {\em Proceedings of the ACM SIGCOMM 2022 Conference}, SIGCOMM '22,
  page 19–35, New York, NY, USA, 2022. Association for Computing Machinery.

\bibitem{246336}
Matthew~K. Mukerjee, Christopher Canel, Weiyang Wang, Daehyeok Kim, Srinivasan
  Seshan, and Alex~C. Snoeren.
\newblock Adapting {TCP} for reconfigurable datacenter networks.
\newblock In {\em 17th USENIX Symposium on Networked Systems Design and
  Implementation (NSDI 20)}, pages 651--666, Santa Clara, CA, February 2020.
  USENIX Association.

\bibitem{valiant1982scheme}
Leslie~G. Valiant.
\newblock A scheme for fast parallel communication.
\newblock {\em SIAM journal on computing}, 11(2):350--361, 1982.

\bibitem{ns3}
{ns-3}.
\newblock {Network Simulator}.
\newblock \url{https://www.nsnam.org/}.

\bibitem{bacharach1966matrix}
Michael Bacharach.
\newblock Matrix rounding problems.
\newblock {\em Management Science}, 12(9):732--742, 1966.

\end{thebibliography}
}

\medskip

\noindent \textbf{{\large Acknowledgments}}

\medskip
\noindent
This work is part of a project that has received funding from the European Research Council (ERC) under the European Union’s Horizon 2020 research and innovation programme, consolidator project Self-Adjusting Networks (AdjustNet), grant agreement No. 864228, Horizon 2020, 2020-2025.
\begin{figure}[!h]
    \centering
    \includegraphics[width=0.6\linewidth]{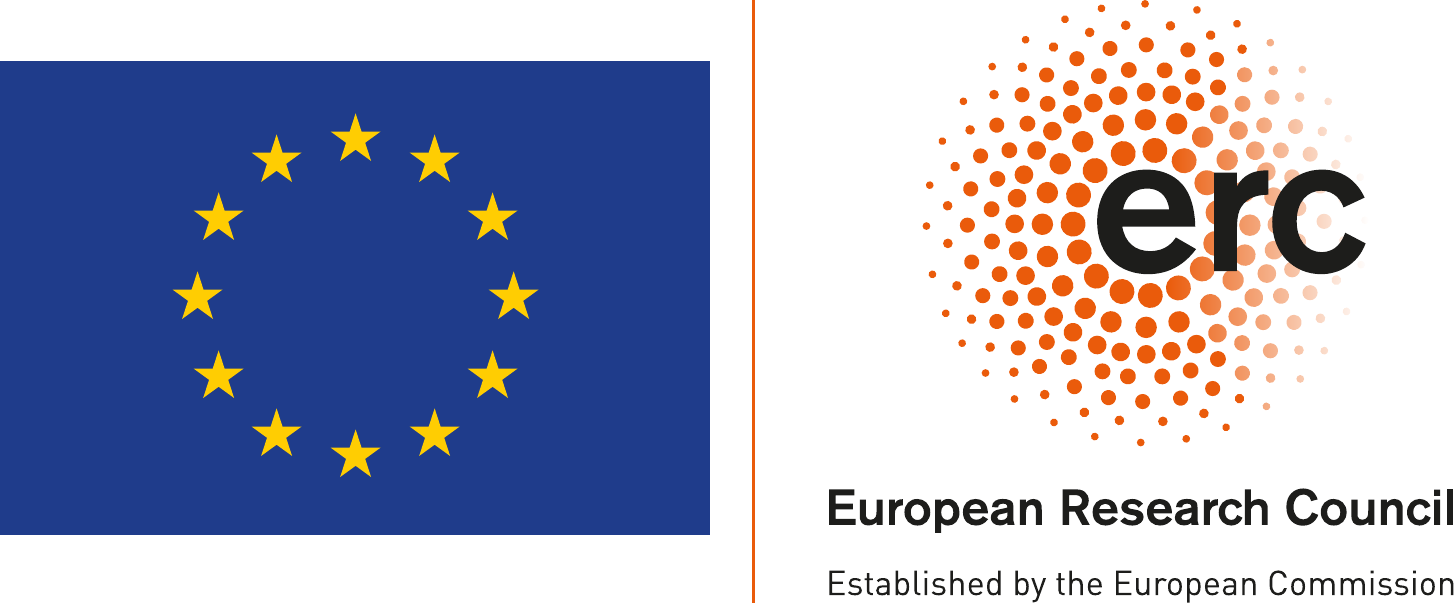}
    % \caption{Caption}
    \label{fig:my_label}
\end{figure}
\newpage
\appendix

\section{Linear Program Formulation}\label{sec:linear-program}
Throughput maximization is a variant of multi-commodity maximum flow problem. For completeness, in the following, we present the linear program (LP) used in \S\ref{sec:eval}.  Given a network of $n$ nodes, the LP takes demand matrix $\mathcal{M}$, and the number of links $\hat{c}^{i,j}$ between each node pair $(i,j)$ as input. Setting link capacities to $1$, we interpret $\hat{c}^{i,j}$ as the capacity between $i,j$. We use $f_{i,j}^{s,d}$ to refer to the flow on edge $(i,j)$ corresponding to $(s,d)$ demand. Our objective is to maximize throughput $\theta$ such that the scaled demand matrix $\theta\cdot \mathcal{M}$ satisfies source-destination demands, flow conservation and capacity constraints. 

\medskip
\noindent\textbf{Input:} Demand matrix $\mathcal{M} = \{m_{s,d} \mid s\in V,\ d\in V\}$
\\
\phantom{}\hspace{30pt}Number of links $\hat{c}^{i,j}$ between $(i,j)$. \\

\noindent \textbf{Objective Function:}
\begin{align*}
\text{Maximize} \quad & \theta
\end{align*}

\noindent \textbf{Subject to the constraints:}
\begin{align*}
\text{Source demand:} & \sum_{i\in V\backslash \{s\}} f^{s,d}_{s,i} \ge \theta \cdot m_{s,d} \\
& \quad\quad\quad\quad\quad\quad\quad \forall s\in V,\ \forall d \in V \\
\text{Destination demand:} & \sum_{i\in V\backslash \{d\}} f^{s,d}_{i,d} \ge \theta \cdot m_{s,d}\\ 
& \quad\quad\quad\quad\quad\quad\quad \forall s\in V,\ \forall d \in V  \\
\text{Flow conservation:} & \sum_{i\in V\backslash\{j\}} f^{s,d}_{i,j} - \sum_{k\in V\backslash\{j\}} f^{s,d}_{j,k} = 0  \\
& \quad\quad\quad\quad\quad\quad\quad\quad\forall j \in V\backslash\{s,d\}\\
& \quad\quad\quad\quad\quad\quad\quad\forall s\in V,\ \forall d \in V\\
\text{Capacity constraints:} & \sum_{s\in V} \sum_{d\in V} f^{s,d}_{i,j} \le \hat{c}^{i,j} \\
& \quad\quad\quad\quad\quad\quad\quad\forall i\in V,\ \forall j \in V
\end{align*}

\noindent \textbf{Variables:}
\begin{align*}
\text{Flow: } & f_{i,j}^{s,d} \ge 0 \ , \ f_{i,j}^{s,d} \in \mathbb{R} \\
& \quad\quad\quad\quad\quad\forall i\in V, j\in V, s\in V,\ \forall d \in V\\
\text{Throughput: } & \theta\ge 0 \ , \ \theta \in \mathbb{R} \\
\end{align*}

\medskip
The above linear program requires the topology as an input. We construct the static network using random regular graphs. We use the complete graph for demand-oblivious networks due to their throughput equivalence~\cite{10.1145/3579312}.

For demand-aware static networks, we first scale the given matrix by a value \texttt{iter}, initially set to $1$. We then decompose the scaled matrix into floor and residual matrices. We interpret the floor matrix as adjacency matrix and add corresponding links between source-destination pairs. After this step, let $d$ be the minimum number of remaining outgoing as well as incoming links available at every node. We construct a $d$-regular random graph using the remaining links to satisfy the residual matrix. We then run the above linear program to maximize throughput $\theta$. We run the above approach for different values of \texttt{iter} from $1$ to $0$ in steps of $0.01$, until the linear program objective reaches $1$. Let $\hat{\theta}$ be the value of \texttt{iter} when the linear program objective reaches $1$. The maximum throughput for the given demand matrix is then $\hat{\theta}$ since we scale the demand matrix by $\hat{\theta}$ and the scaled matrix is feasible in the network (objective $=1$).

For demand-aware periodic networks, we use Corollary~\ref{cor:static-periodic} to  construct the equivalent demand-aware static network and then compute throughput using the heuristic described above.

We note that the iterative approach used in our heuristic has an error margin of $0.01$ since we only iterate in steps of $0.01$. Given the significant differences in the throughput of different types of networks, we consider our error margin to be negligible. Further, our integer-residual decomposition technique may not always result in the optimal topology. Alternatively, computing the optimal topology can be incorporated in the throughput maximization linear program by using $\hat{c}^{i,j}$ as integer variables, subject to degree constraints. However, this approach turns out to be impractical even for small network sizes. In fact, prior work resorted to maximum link utilization as an objective~\cite{10.1145/3544216.3544265}. Developing performant approximation algorithms for the throughput maximization problem in demand-aware networks is still an open area of research.

\section{Deferred Proofs}\label{sec:deferred-proofs}

\begin{proof}[Proof of Corollary~\ref{cor:throughput-periodic-lb}]
Since the throughput of a demand-aware static network of degree $n$ with link capacity $\frac{c\cdot u}{n}$ is lower bounded by $\frac{2}{3}$ for uniform-residual demand matrices, proof follows from Corollary~\ref{cor:static-periodic} for the throughput of a demand-aware periodic network with degree $u$ and link capacity $c$.
\end{proof}

\begin{proof}[Proof of Corollary~\ref{cor:throughput-periodic-ub}]
Since the throughput of a demand-aware static network of degree $n$ with link capacity $\frac{c\cdot u}{n}$ is at most $\frac{4}{5}$ from Theorem~\ref{th:ub-static}, the throughput of a demand-aware periodic network cannot be greater than $\frac{4}{5}$ due to Corollary~\ref{cor:static-periodic}.
\end{proof}

\newpage

\onecolumn

\section{Additional Results}

\begin{figure}[!h]
    \center
    \begin{subfigure}{1\linewidth}
        \centering
        \includegraphics[trim={0.7cm 0.5cm 0 1.2cm},clip,width=0.9\linewidth]{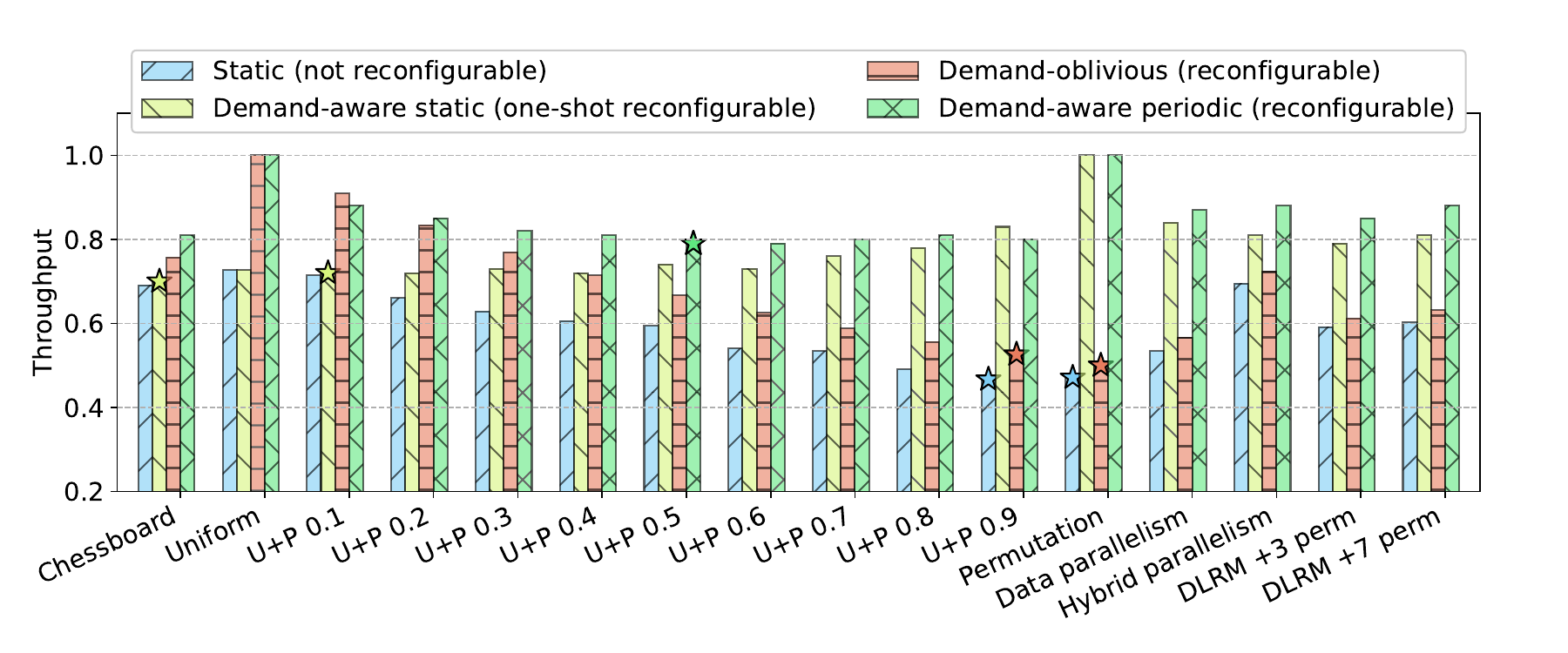}
        % \vspace{-4mm}
        \subcaption{Degree$=8$}
    \end{subfigure}
    \begin{subfigure}{1\linewidth}
        \centering
        \includegraphics[trim={0.7cm 0.5cm 0 1.2cm},clip,width=0.9\linewidth]{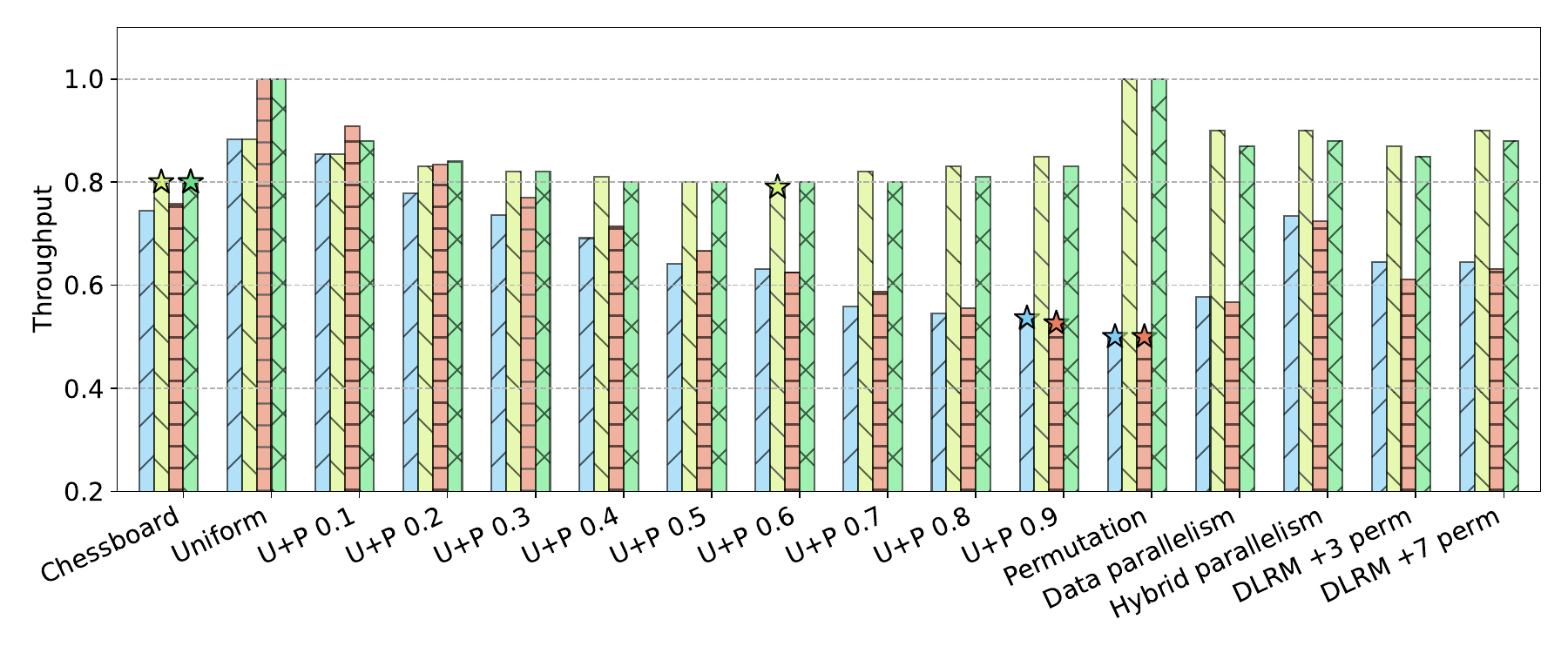}
        % \vspace{-4mm}
        \subcaption{Degree$=12$}
    \end{subfigure}
    \begin{subfigure}{1\linewidth}
        \centering
        \includegraphics[trim={0.7cm 1.35cm 0 1.2cm},clip,width=0.9\linewidth]{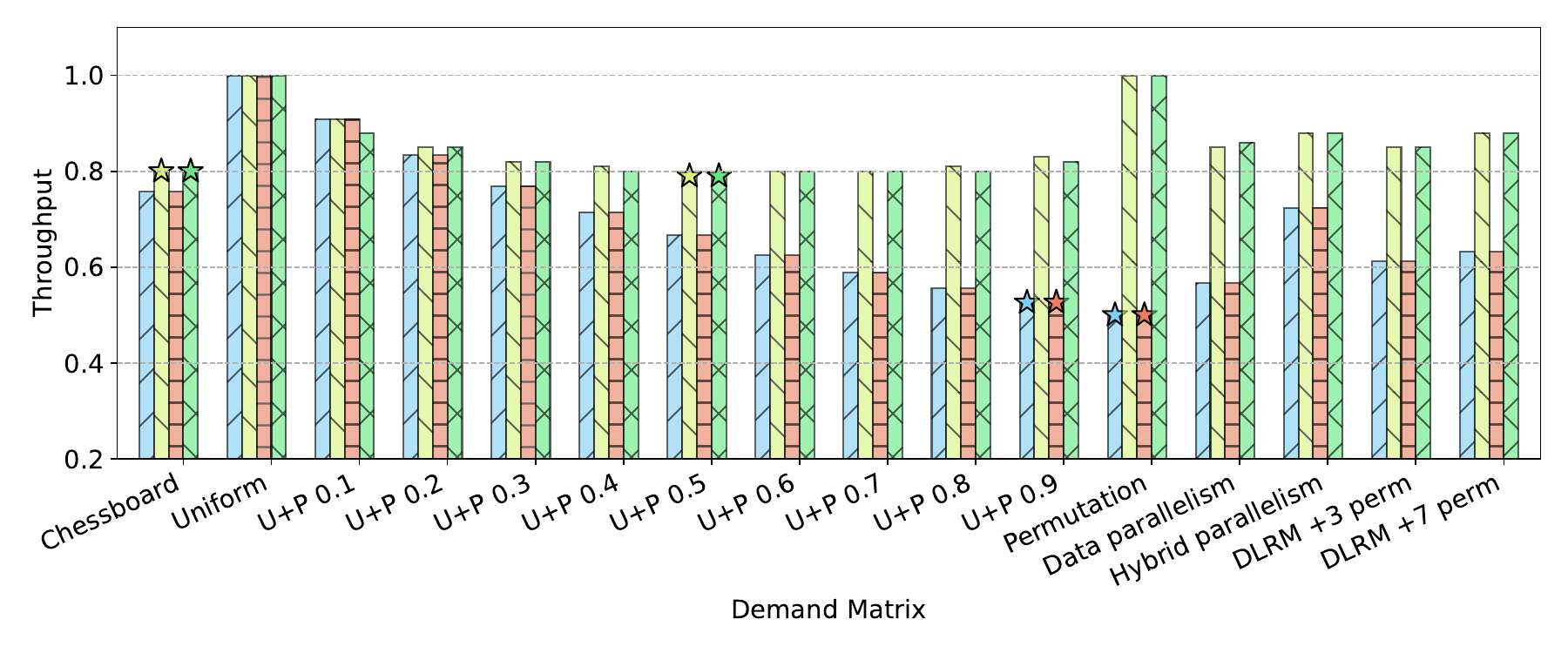}
        % \vspace{-4mm}
        \subcaption{Degree$=16$}
    \end{subfigure}
    \caption{Throughput of reconfigurable datacenter networks under various demand matrices for different degrees \ie the number of incoming and outgoing links of the physical topology.}
    \label{fig:degrees-throughput}
\end{figure}

\label{LastPage}
\end{document}